\documentclass[draftclsnofoot,onecolumn,12pt]{IEEEtran} 

\usepackage{amsmath,amssymb,amsthm}
\usepackage{graphicx}

\newcommand{\probln}[1]{p_{#1}}
\usepackage{bbm}
\usepackage{amssymb}
\usepackage{amsmath,amsthm}
\usepackage{color}
\usepackage{amsfonts}
\usepackage{graphicx}
\usepackage{verbatim}
\usepackage{epstopdf}
\usepackage{cite}
\usepackage{tabulary}
\usepackage{mathtools}
\usepackage{enumitem}

\newcommand{\tP}{p_\tx}

\newcommand{\iftext}[1]{\text{#1}}

\newcommand{\noise}{\sigma_n^2}

\newcommand{\dist}[1]{\|#1\|}

\newcommand{\origin}{\mathbf{o}}

\newcommand{\palmprobx}[2]{\mathbb{P}^{#2}{\left[#1\right]}}

\newcommand{\palmexpectx}[2]{\mathbb{E}^{#2}{\left[#1\right]}}

\newcommand{\dd}{\mathrm{d}}

\newcommand{\fracS}[2]{#1/#2}
\newcommand\expect[1]{\mathbb{E}\left[#1\right]}
\newcommand\prob[1]{\mathbb{P}\left[#1\right]}

\newcommand\indside[1]{\mathbbm{1}\left({#1}\right)}

\newcommand{\SINR}{\mathsf{SINR}}

\newcommand{\Rate}{\mathsf{Rate}}

\newcommand{\Ball}{\mathcal{B}}
\newcommand{\expects}[2]{\mathbb{E}_{#1}\left[#2\right] }
\newcommand{\probs}[2]{\mathbb{P}_{#1}\left[#2\right] }
\newcommand{\laplace}[1]{\mathcal{L}_{#1} }
\newcommand{\ie}{{\em i.e.} }

\newcommand{\Pc}{\mathrm{p_c}}
\newcommand{\Rc}{\mathrm{r_c}}

\newcommand{\SThres}{\tau}

\renewcommand{\t}{\mathbf{t}}

\newcommand{\tx}{\mathsf{t}}

\renewcommand{\L}{\mathrm{L}}
\newcommand{\N}{\mathrm{N}}

\newcommand{\y}{{\mathbf{y}}}
\newcommand{\x}{{\mathbf{x}}}
\newcommand{\expU}[1]{e^{#1}}
\newcommand{\X}{\mathbf{X}}
\newcommand{\Y}{\mathbf{Y}}
\newcommand{\expS}[1]{\exp{\left(#1\right)}}

\def\home{\hbox{\kern3pt \vbox to13pt{}% 
   \pdfliteral{q 0 0 m 0 5 l 5 10 l 10 5 l 10 0 l 7 0 l 7 5 l 3 5 l 3 0 l f
               1 j 1 J -2 5 m 5 12 l 12 5 l S Q }%
   \kern 13pt}}

\newtheorem{theorem}{Theorem}
\newtheorem{lemma}{Lemma}

\newtheorem{remark}{Remark}

\newcommand{\ue}{\mathrm{u}}
\graphicspath{{FigFiles/},{Blockage1D/}}
\allowdisplaybreaks

\IEEEoverridecommandlockouts
\newcommand{\sE}{\mathsf{E}}
\newcommand{\qr}{q}
\newcommand{\ql}{q'}
\newcommand{\Q}{Q}
\newcommand{\Pci}[1]{\mathrm{p}_{\mathrm{c}#1}}
\newcommand{\LAOP}{LAP~}
\newcommand{\isdefined}{\stackrel{\Delta}{=}}
\newcommand{\IB}{\mathsf{\ I}}
\renewcommand{\x}{{\boldsymbol{x}{}}}

\newtheorem{cor}{Corollary}

\renewcommand{\L}{\mathsf{L}}
\renewcommand{\N}{\mathsf{N}}
\newcommand{\F}{F}
\newcommand{\s}[1]{s_{#1}}
\newcommand{\probserv}[1]{g_{#1}}
\newcommand{\exc}[1]{e_{#1}}
\newcommand{\cgain}[1]{C_{#1}}
\renewcommand{\v}{v} %type of interferers
\renewcommand{\u}{s} %type of serving link
\renewcommand{\nu}{z}
\newcommand{\Lfunc}{\kappa}
\newcommand{\excNL}{e_1}
\newcommand{\excLN}{e_2}

\renewcommand{\excNL}{\exc{\N|\L}}
\renewcommand{\excLN}{\exc{\L|\N}}

\newcommand{\excNN}{\exc{\N|\N}}
\newcommand{\excLL}{\exc{\L|\L}}

\title{Impact of Blocking  Correlation on the Performance of mmWave Cellular Networks}
\author{Saurabh Kumar Gupta,  %~\IEEEmembership{Member,~IEEE}, 
	Vikrant Malik, 
		 Abhishek K. Gupta, 
		 and Jeffrey G. Andrews
		\thanks{
S. Gupta, V. Malik and A. Gupta are with Indian Institute of Technology Kanpur, India 208016. % (Corresponding Author: Abhishek Gupta, 
Email:{ saurabhg20@iitk.ac.in, vikrant@iitk.ac.in, gkrabhi@iitk.ac.in}. %). 
J. Andrews is with the University of Texas at Austin, TX 78705, USA. Email: {jandrews@ece.utexas.edu}. 
The support of  Science and Engineering Research Board (Grant SRG/2019/001459) and IITK (Grant EE/2017/157) is greatly acknowledged. Part of this paper was presented at {GLOBECOM} 2020 \cite{GupGC2020}.
		}% <-this % stops
}% <-this % stops a space

\begin{document}

% \vspace{-.1in}\ \\

\maketitle
%	\pagestyle{empty}
%	\thispagestyle{empty}
%\thispagestyle{empty}
%\pagestyle{empty}

%%%%%%%%%%%%%%%%%%%%%%%%%%%%%%%%%%%%%%%%%%%%%%%%%%%%%%%%%%%%%%%%%%%%%%%%%%%%%%%%
\begin{abstract}
   % The past analytical works on the performance of mmWave cellular networks  have assumed  blocking of each link to be independent of other links. In reality, a large blockage can block multiple links. Therefore, there exists a spatial correlation in the blocking of links from a user to different base stations (BS). 
In mmWave networks, a large or nearby object can obstruct multiple communication links, which results in spatial correlation in the blocking probability between a user and two or more base stations (BSs). % which is not studied in the past analytical works. 
    This paper characterizes this blocking correlation and derives its impact on the signal-to-interference-plus-noise ratio (SINR)  of a mmWave cellular network.  %Owing to differences in the analysis of a 1D and 2D mmWave cellular network, we analyze both networks. We, first, derive the distribution of  the serving BS distance from a typical user. 
    We first present an exact analysis of a 1D network and highlight the impact of blocking correlation in the derived expressions. Gaining insights from the 1D analysis, we develop an analytical framework for a 2D network where we  characterize the sum interference at the user by considering the correlation between the blocking of serving and interfering links.   Using this, we derive the SINR coverage probability. Via simulations, we demonstrate that including blockage correlation in the analysis is required for accurate characterization of the system performance, in particular when the blocking objects tend to be large. 
% and compute the probability that the user is associated with a LoS or NLoS BS.  
   % We have derived the equations for the distribution of distance to the serving base station for the line of sight and non-line of sight cases. We have also calculated the Laplace functional of interference for correlated blocking case. Laplace Functional is later used in calculating the SINR distribution of the system. In the last section, we have focused on simplifying the expressions under some realistic assumption. One of the key assumptions is that only a line of sight base station can get connected to the user equipment. Under this assumption and some mathematically derived bounds, we have calculated the bounds on these distributions. The key conclusion is that first-order correlation consideration in the blockage probability gives a lower bound the distance distribution of the cellular system. First order correlation assumption gives a close approximation in the case of large blockage length for SINR distribution.
\end{abstract}

%%%%%%%%%%%%%%%%%%%%%%%%%%%%%%%%%%%%%%%%%%%%%%%%%%%%%%%%%%%%%%%%%%%%%%%%%%%%%%%%
\section{Introduction}
Communication at high frequencies, including mmWave and THz frequencies, can be severely degraded due to obstacles such as buildings, foliage, and humans, collectively known as \emph{blockages} \cite{AndrewsmmWaveTut2016, TriSabGupDhi2021}. Due to the high sensitivity of the link quality to blockage at these frequencies, correct modeling of blocking is essential to accurate analysis of these systems. Moreover, since the same object could block several paths \cite{GupAndHea18ICC}, correlation of blocking events is a key part of such a model. %and thus need to be included in the analysis. 

% %Blockages – in the form of buildings,
%foliage, and people, can severely impact the performance
%of cellular networks by reducing the signal strength and thus
%SNR. Blocking’s effect is more severe at higher frequencies including mm-wave, due to higher penetration losses
%and reduced diffraction [1], [2].  To overcome blockage effects, macrodiversity can
%be leveraged whereby a user is connected to multiple BSs
%simultaneously, which clearly increases the chance for a LOS
%connection [5], [6]. The objective of this
%
%paper is to study the impact of blockages on the gai

\subsection{Related Work}  
Various analytical approaches to modeling  blockages in  cellular networks have been proposed. % Simulation based approaches to model blockages by using ray tracing [19] in a deterministicenvironment are numerically complex and not tractable. 
Perhaps the most popular way to include blocking in the analysis, particularly for sub-mmWave frequencies, is by modeling it as a log-normal random variable. However, this simplistic approach misses many crucial features, such as the dependence of blocking probability on the link distance or a change in the effective pathloss exponent due to blocking. % over its blocking probability. Therefore, this approach may not be not suitable to analyze scenarios where blockages play a significant role in determining its performance which is
%the case with communication at higher frequencies including mm-wave. 
In recent works, particularly those based on stochastic geometry, a Poisson Boolean process \cite{BaccelliBook} has been used to model the set of random blockages. In this model, blockages are seen as random objects (e.g. in the form of linear segments \cite{BaiLineBlock2012} or other shapes including rectangular, polygonal, or circular  \cite{bai2014coneanalysis}) distributed in the space according to a Poisson point process (PPP). This model provides us with a link blocking probability expression which can be incorporated to investigate the impact of blockages on the SINR coverage of mmWave systems  \cite{bai2014coverage,bai2014analysis}. 
The framework was later extended to analyze various mmWave networks \cite{AndrewsmmWaveTut2016} including multiple tiers \cite{Ren2015}, spectrum sharing \cite{GupAndHea2016} and self-backhauled networks \cite{SinKulGhoAnd2015}.

All the aforementioned works assume that blocking events are independent of each other. However, in reality, the blocking events can be correlated. For example, a large blockage can block downlink signals from multiple BSs to a given user. This effect is more severe in a one dimension (1D) network.  For example, in a  quasi-1D deployment of BSs on an urban street, a single blockage will block all BSs behind it. 
In our past work \cite{gupta2017macrodiversity}, we characterized blocking correlation  in a 2D mmWave network in the presence of  linear blockages with random lengths and orientations. In particular, we derived the joint probability of two links to be  line-of-sight (LoS)  as a function of their lengths and the angle between them.  Using a similar blocking correlation model, the empirical SINR assuming a fixed number of blockages was computed in \cite{HriVal2019}. The blockage correlation was also studied in \cite{AdiDhiMolBeh2018}  for a sensor localization problem. Similar to spatial correlation, time correlation of blocking events caused by  the user mobility at two time instants was studied in \cite{Sam2016}. In \cite{baccelli2015correlated}, the authors studied the correlation in the interference between two locations of a user. However, the correlation due to simultaneous blocking of BSs was not considered.

The blocking correlation can have a significant effect on the performance of a mmWave cellular network. For example, the presence of a  LoS serving link may increase the probability of having LoS interferers, degrading the system performance. Similarly, the absence of an LoS serving link reduces the chance of having an LoS interferer, which improves the system performance. However, there exists no work which analyzes the impact of this correlation on the SINR or rate coverage of cellular networks, which is the main focus of this work.

\subsection{Contributions}
   In this paper, we study a mmWave network in the presence of random blockages, which are modeled using a Boolean Poisson process. In contrast to  previous works, we consider the correlation among blocking events occurring in links while deriving the performance of the network. We first present an exact analysis for a 1D network. Leveraging the qualitative and quantitative insights obtained from the 1D case, we achieve an approximate analysis for the 2D network. In particular, this work has the following contributions:
   \begin{enumerate}
       \item
       We analyze a 1D mmWave cellular network, corresponding to a street or highway.  For a user, the blocking states of all BSs depend on only two variables that are the locations of the nearest blockage on both sides of the user.  We calculate the distribution of the distance of the serving BS from the typical user conditioned on these variables. We, then derive the Laplace transform of the interference at the typical user, taking the blocking correlation into account. Using these results, we derive the exact SINR coverage probability of the network.
       \item  
       We analyze the more general 2D deployment of a cellular network. Exact analysis is intractable, so we consider the first order blocking correlation.  We first characterize the conditional probability of having a LoS and NLoS interferer at a particular location conditioned on the blocking state of serving link.  We  derive the distribution of distance to the serving BS from a typical user and compute the probability that the user is associated with a LoS or NLoS BS.  We then derive the Laplace transform of the interference at the user by considering the blocking correlation between the serving and interfering links. Using these results and approximations, we provide compact expressions for the SINR coverage probability, including simplified bounds, and the rate coverage expression.
       \item
       We present numerical results to demonstrate the impact of including blocking correlation on the analysis and show its importance in the accurate characterization of the system performance.  In particular, the effect can be severe when the blockages tend to be large objects. 
\end{enumerate}

\textbf{Notation}:  $x$ denotes the distance of a location $\x$ from the origin. $\laplace{X}(s)$ denotes the Laplace transform (LT) of a random variable (RV) $X$. $\palmprobx{\cdot}{\x}$ denotes the probability under the Palm measure \ie the probability conditioned on the occurrence of a point at $\x$. $\IB_n$ is the modified Bessel function of the first kind of the order $n$.

\section{System Model}
\label{chap:sys_model}

In this paper, we consider a mmWave cellular network in the presence of random blockages. The system model is as follows.

\textbf{Network model}: We consider that BSs are deployed as a homogeneous PPP 
$\Phi$ with density $\lambda$ in the space $\sE=\mathbb{R}^d$ for $d=1,2$. Due to inherent differences in their analysis, we consider 1D  and 2D  deployments separately in the  next two sections. We denote the location of $i$th BS by $\X_i$. The user process is assumed to stationary and independent of the BS process.  Its density is $\lambda_\ue$. Due to stationarity, we can take a typical user at the origin.  Let the distance of $i$th BS from this user is given as $X_i=\|\X_i\|$. The system bandwidth available at each BS is $B$.

\textbf{Blockage model}: Blockages are distributed as a Poisson Boolean process \cite{gupta2017macrodiversity} where blockages are in the form of line segments with their centers located according to a homogeneous PPP $\Psi$ of density $\mu$ in $\sE$. Let $L$ and $\Delta$ represent random variables denoting the length and orientation of a typical blockage. $L$ and  $\Delta$'s distributions are denoted by $F_{\mathrm{L}}(l)$ and $F_{\mathrm{\Delta}}(\delta)$. Let us consider a link between the $i$th BS and the user. If there is no blockage  falling on the link between the BS and the user, the signal received at the UE will be dominated by LoS component. Such a link is known as LoS link. In case of any blockage falling on the link between the BS and the user, the LoS component of the signal gets blocked by the blockage and BS will be non-line-of-sight (NLoS). Considering the randomness of blockages,  the event that this link is LoS can be characterized by a Bernoulli random variable with the parameter denoting LoS probability of the link, similar to \cite{bai2014analysis}.  
Let $\s{i}$ denote the blocking state (or type) of the link. Here, $\s{i}\in\{\L,\N\}$. The LoS probability of a link depends on its link-distance  and the exact behavior depends on the blockage environment. For the linear blockage model considered in this paper, 
the probability that a link with link-distance $r$ is LoS is given as
   $ \probln{\L}(r) = \exp(-\beta r)$
where the blockage parameter $\beta$ depends on the blockage process. Similarly, the probability that a link is NLoS is 
   $ \probln{\N}(r) = 1-\exp(-\beta r).
$

\textbf{Path-loss model}: 
In this paper, we  consider a general path-loss model where the power received at the user from a LoS BS and NLoS BS located at distance $x$ from it, is given by $\ell_{\L}(x)$ and $\ell_{\N}(x)$ respectively. The path-loss function satisfies the following  realistic conditions:
\begin{enumerate}
\item $\ell_{\L}(x) \ge \ell_{\N}(x) \ \forall x$. 
\item  $\ell_{\L}(x_1) \ge \ell_{\L}(x_2)$ and $\ell_{\N}(x_1) \ge \ell_{\N}(x_2) \  \forall x_1 < x_2$.
\item $C_\L = \max_x (\ell_{\L}(x)) <\infty $ and $C_\N = \max_x (\ell_{\N}(x))<\infty$.
\end{enumerate}
The first condition means that the power received from a LoS BS can not be smaller than that of a NLoS BS, located at the same distance from the user. The second condition means that  the received power  decreases with $x$. The third  condition implies that the path-loss is bounded. The transmit power $\tP$ of BS is included in the path-loss function.

 One special case is the {\em  standard bounded power law poth-loss} (BPLP) model where  $\ell_{\L}(x) = \min(C_\L, C_\L x^{-\alpha_\L})$ and $\ell_{\N}(x) = \min(C_\N, C_\N x^{-\alpha_\N})$, where $\alpha_s$ is the path-loss exponent and $C_s$ is the near-field gain for a link of type $s$ with $C_\L>C_N$. 
We also denote $\alpha = \frac{\alpha_\L}{\alpha_\N}$ and $c = \left(\frac{C_\N}{C_\L}\right)^{\frac{1}{\alpha_\N}}$.

We also consider another special case where NLOS links are in the complete outage \ie $\ell_\N(r)=0$ for all $r$. We call it {\em LOS association only path-loss} (LAP) model.

For simplicity, we consider Rayleigh fading which is independent across link. The fading corresponding to $i$th BS is $H_i$ which is distributed as $\mathsf{Exp}(1)$. Hence, the received power from the $i$th BS located at $\X_i$ is $P_i=H_i \ell_{s_i}\left({X_i}\right).$

\textbf{Cell-association model}: 
We consider the association based on the maximum average received power. In this criterion, a user connects to the BS providing  highest received power at the user averaged over fading which can be LoS or NLoS. For the typical user, the serving BS is termed the {\em tagged} BS and denoted by index $0$. Hence, the serving power is $S=P_0=H_0 \ell_{s_0}\left({X_0}\right)$. The rest of the BSs contribute to the  interference. Note that according to the  association rule, the user receives less power from each interfering BS compared to the serving BS \ie $P_i<S$. Therefore, condition on the serving BS's location creates an exclusion ball around the user where these BSs cannot be located \cite{GupAndHea2016}. The radius of this ball is different for LoS and NLoS BSs and it also depends on the type of the serving BS. If serving BS is LoS, the radius for LoS and NLoS BSs are $\excLL{}(X_0)$ and $\excNL{}(X_0)$ where 
\begin{align}
\excLL{}(x)=x, && \excNL{}(x)=\min_y \{y:\ell_\N(y) \le \ell_\L(x)\}.
\end{align}
Similarly,  if the serving BS is NLoS, the exclusion radius for LoS and NLoS BSs are $\excLN{}(X_0)$ and $\excNN{}(X_0)$   with 
\begin{align}
\excLN{}(x)=\min_y \{y:\ell_\L(y) \le \ell_\N(x)\}, && \excNN{}(x)=x
\end{align}
Further it can be shown that $
\excNL{}(x) < x$ and $ \excLN{}(x) > x$. 

Note that for the special case BPLP model, these functions are given as
\begin{align}
\excNL{}(x) \!= \!
    \begin{cases}
    0 &\!\!\!\! \text{if } x \leq \ell_{\L}^{-1}(C_\N) =c^{-\frac1\alpha}\\
    \ell_{\N}^{-1}(\ell_{\L}(x))=cx^\alpha &\!\!\!\! \text{if }x > \ell_{\L}^{-1}(C_\N)=c^{-\frac1\alpha}
    \end{cases},
    &&
    \excLN{}(x)\!=\!
    \begin{cases}
     \ell_{\L}^{-1}(\ell_{\N}(x))=\left(\frac{x}c\right)^{\frac1\alpha}&\!\!\! \text{if }x>1.\\
     \left(\frac{1}c\right)^{\frac1\alpha}&\!\!\! \text{if }x\le 1.
    \end{cases} \label{eq:excl}
\end{align}

\textbf{SINR}: 
The sum interference at the typical user is given as
\begin{align*}
I=\sum\nolimits_{\X_i\in\Phi, \X_i\ne \X_0}P_i=\sum\nolimits_{\X_i\in\Phi\setminus\{ \X_0\}}H_i \ell_{s_i}\left({X_i}\right).
\end{align*}%
We can split the $\Phi$ into two point process $\Phi_\L$ and $\Phi_\N$ denoting the sets of interfering LoS and NLoS BSs respectively. Hence, $I$ can be written as  
\begin{align*}
&I=I_\L+I_\N & \text{with } I_\v=\sum\nolimits_{\X_i\in\Phi_\v, \X_i\ne \X_0}H_i C_{\v}\dist{\X_i}^{-\alpha_{\v}}
\end{align*}
where  $I_\v$ is the interference from BSs of type $\v$. Hence, the SINR at the typical user is given as
\begin{align}
\SINR=\frac{S}{\noise+I},
\end{align}
where $\noise$ is the noise. 

We will use two performance metrics: the SINR coverage  $\Pc$ and  rate coverage $\Rc$ which are the complimentary CDFs of the SINR and the rate at the typical user and are given as
 \begin{align}
 \Pc(\SThres)=\prob{\SINR>\SThres},&&
	\Rc(\rho) = \prob{\Rate > \rho}.\label{eq:pcdef}
\end{align}
In the next section, we present the exact analysis of the 1D case.
 
 \section{1D Cellular Network}
 We will first consider a 1D mmWave cellular network such as one deployed on a straight street \ie the  space $\sE$ is the real line $\mathbb{R}=(-\infty, \infty)$. Such a deployment can be seen in dense urban environments, {\em e.g.} downtown of a city \cite{GhaMalKalGup2021} or on a highway with vehicular users.  In 1D, blockages can be seen as points on  $\mathbb{R}$. A link of length $r$ is LoS if there is no blockage in this line-segment. This occurs with probability $\probln{\L}(r)=e^{-\mu r}$ which means that $\beta=\mu$. A single blockage  shadows all BSs falling beyond its location. Therefore, from a user's perspective, if a BS is blocked, then all BSs falling behind it are also blocked. This results in a severe correlation among blocking events of links between a user and different BSs. Therefore, the independent blocking assumption is clearly invalid in this scenario.

Let the closest blockage on the right side of the typical user (\ie in $\mathbb{R}^+=(0,\infty)$) be at distance $\qr$. Similarly, $\ql$ denotes the distance of the closest blockage on the left side (\ie in $\mathbb{R}^-=(-\infty,0)$).  Both $\qr$ and $\ql$ have the same distribution with probability density function (PDF),
\begin{align}
    f_\Q(q) =  \mu \expS{- \mu q}, \quad q \ge 0
    \label{eq:distb}
\end{align}
\ie $\qr,\ql\sim\mathsf{Exp}(1/\mu)$. Note that locations of only these two blockages will decide for the blocking state of the link to every BS. In particular, all the BSs falling on the right side of the user with $|\X_i|>\qr$ will be NLoS and all the BSs falling on the right side of user with $|\X_i|<\qr$ will be LoS. Therefore, conditioned on $\ql$ and $\qr$,   $\Phi_\L$ (PPP of LoS BSs) is given as $ \Phi_\L=\Phi^+_\L \cup \Phi^-_\L$ where  
\begin{align}
& \Phi_\L^+=\{\X_i: \X_i\in\Phi, \X_i>0, \X_i<\qr\}, \text{
and }
\Phi_\L^-=\{\X_i: \X_i\in\Phi, \X_i<0, |\X_i|<\ql\}\end{align}
denotes right side and left side LoS BSs respectively. Similarly, 
 $\Phi_\N=\Phi^+_\N \cup \Phi^-_\N$ with 
\begin{align}
& \Phi_\N^+=\{\X_i: \X_i\in\Phi, \X_i>0, \X_i>\qr\}, \text{
and }
\Phi_\N^-=\{\X_i: \X_i\in\Phi, \X_i<0, |\X_i|>\ql\}.
\end{align}

To compute the SINR coverage probability of the typical user, we will need the distribution of the distance of the serving BS from the user.  This is derived next.

\subsection{Distance Distribution of LoS and NLoS Serving BS}

Let $E_{\L,\x}$ denote the event that the BS at the location $\x$ is LoS and serving the typical user. Similarly, $E_{\N,\x}$ denotes the event that the BS at the location $\x$ is the NLoS  serving BS.  We first provide the probabilities of these two events conditioned on  $\ql$ and $\qr$ in the following Lemma.

\begin{lemma}\label{lem:gs1DConditionedqq}
Conditioned on the locations of the closest blockages $\qr$ and $\ql$, the probability that the typical user in a 1D mmWave cellular network is associated with a LOS (and NLOS) BS located at $\x$ is given as
\begin{align}
&{g_\L(x | \qr, \ql)}  \isdefined \palmprobx{
        E_{\L,x}|\ql,\qr
         }{\x}
         \!
         =\!
    \begin{cases}
    \mathbbm{1}(x < q) \exp(-\lambda x)\exp\left(-\lambda \excNL{}(x) \right) &\!\! \text{if }q' \le  \excNL{}(x) \\
    \mathbbm{1}(x < q) \exp (-\lambda x) \exp \left(-\lambda q^{\prime}\right) &\!\! \text{if }\excNL{}(x)\le q^{\prime}<x \\
    \mathbbm{1}(x < q) \exp (-2 \lambda x) & \!\! \text{if }   x \le \ql
    \end{cases}\label{eq:cond_gL}
    \\
 &   g_\N(x|\qr,\ql)   \isdefined\! \palmprobx{        E_{\N,x}|\ql,\qr         }{\x}\!
 =\!\!
    \begin{cases}
    \mathbbm{1}(x \geq q) \exp(-2\lambda x) & \text{if }q' <  x \\
    \mathbbm{1}(x \geq q) \exp (-\lambda x) \exp \left(-\lambda q^{\prime}\right) & \text{if }x < q^{\prime} < \excLN{}(x) \\
    \mathbbm{1}(x \geq q) \exp (-\lambda x) \exp (-\lambda \excLN{}(x)) & \text{if }q' > \excLN{}(x)
    \end{cases}.
    \label{eq:cond_gN}
\end{align}
\end{lemma}
\begin{IEEEproof}
See Appendix \ref{appendix:proofth1}
for the proof.
\end{IEEEproof}

We now derive the  distribution of the distance to the LoS serving BS in terms of  the joint probability  that the serving BS is at a distance larger than $r$ and is LoS. It is given in the following Theorem.

\begin{theorem}\label{thm:GLr1D}
The joint probability $G_{\L}(r)$ that the serving BS is outside the ball $\Ball(\origin,r)=[-r,r]$ and is LoS to the typical user in a 1D mmWave cellular network is  given as
\begin{align}
        &G_{\L}(r)=
        2\lambda \int_{r}^\infty g_\L(x)  \dd x
         \label{eq:GLR1D_th} \text{ with }\\
&g_\L(x)\isdefined \palmprobx{E_{\L,\x}}{\x}=\expects{\ql,\qr}{g_\L(x|\ql,\qr)}
\label{eq:gLx1D_middlestep1}
\\
&= e^{-\lambda\left(\excNL{}(x)+x\right)-\mu x} \left(1-e^{-\mu \excNL{}(x)}\right) + \frac{ \mu e^{-(\lambda + \mu) x}}{\lambda+\mu}  \left[e^{-(\lambda+\mu)\excNL{}(x)}-e^{-(\lambda+\mu) x}\right] 
	+  e^{-2 (\lambda + \mu) x}.
	\label{eq:gLx}
\end{align}

\end{theorem}
\begin{IEEEproof}
See Appendix \ref{app:thm:GLr1D}.
\end{IEEEproof}

% ---------------------------------------------------------------- %
Similarly, for the distribution of the distance to the NLoS serving BS, we have the following Theorem.

\begin{theorem}\label{thm:GNr1D}
 The joint probability  $G_{\N}(r)$ that the serving BS is outside $\Ball(\origin,r)$ and is NLoS to the typical user in a 1D mmWave cellular network is  given as
\begin{align}
        G_{\N}(r)&=
        2\lambda \int_{r}^\infty g_\N(x)  \dd x \text{ with } \\
        g_\N(x)  & \isdefined \palmprobx{E_{\N,\x}}{\x}=
        \expects{\ql,\qr}{g_\N(x|\ql,\qr)}    \label{eq:GNR1D_1}\\
        &= \frac{\mu e^{-\lambda x}}{\lambda+\mu}\left(1-e^{-\mu x}\right)\left[e^{-(\lambda+\mu) x}-e^{-(\lambda+\mu)\excLN{}(x)}\right]  \nonumber \\
& \qquad+e^{- \lambda x}e^{-(\lambda+\mu)\excLN{}(x)}\left(1-e^{-\mu x}\right) +  e^{-2 \lambda x}\left(1-e^{-\mu x}\right)^{2}.
        % \dd x   
         \label{eq:gNx}
        \end{align}
\end{theorem}

\begin{cor}
For the special case BPLP model, $g_{\L}(x)$ and $g_{\N}(x)$ are given as
\begin{align}
&g_\L(x) =
\begin{cases}
 \frac{ \mu}{\lambda+\mu} e^{-(\lambda + \mu) \x} \left[e^{-(\lambda+\mu)c \x^{\alpha}}-e^{-(\lambda+\mu) \x}\right] \\
 \hspace{1.01in}+ e^{-\lambda\left(c \x^{\alpha}+\x\right)} e^{-\mu \x} \cdot\left(1-e^{-\mu c \x^{\alpha}}\right) +  e^{-2 (\lambda + \mu) \x}& \iftext{if } x > {c}^{-\fracS{1}{\alpha}}\\
  \frac{ \mu}{\lambda+\mu} e^{-(\lambda + \mu) \x} \left[1-e^{-(\lambda+\mu) \x}\right] + e^{-2 (\lambda + \mu) \x}&\iftext{if } x \le {c}^{-\fracS{1}{\alpha}}
 \end{cases}\\
&g_{\N}(x) = 
\begin{cases}
\frac{\mu e^{-\lambda x}}{\lambda+\mu}\left(1-e^{-\mu x}\right)\left[e^{-(\lambda+\mu) x}-e^{-(\lambda+\mu)\left(\frac{x}{c}\right)^{\frac{1}{\alpha}}}\right] + \nonumber \\ 
 \hspace{.5in} +e^{- \lambda x}e^{-(\lambda+\mu)\left(\frac{x}{c}\right)^{\frac{1}{\alpha}}}\left(1-e^{-\mu x}\right) +  e^{-2 \lambda x}\left(1-e^{-\mu x}\right)^{2},&\iftext{if } x>1 \\
 \frac{\mu e^{-\lambda x}}{\lambda+\mu}\left(1-e^{-\mu x}\right)\left[e^{-(\lambda+\mu) x}-e^{-(\lambda+\mu){c^{-\fracS{1}{\alpha}}}}\right] + \nonumber \\
  \hspace{.5in}+  e^{- \lambda x}e^{-(\lambda+\mu)%\left(\frac{1}{c}\right)^{\frac{1}{\alpha}}
 c^{-1/\alpha}
 }\left(1-e^{-\mu x}\right) +  e^{-2 \lambda x}\left(1-e^{-\mu x}\right)^{2},& \iftext{if } x \le 1
\end{cases}.
\end{align}
\end{cor}

\begin{remark}
Note that under the independent blocking assumption, $g_\L(x)$ and $g_\N(x)$ are given by
\begin{align}
\overline{g_{\L}}(x) &= e^{-\mu x}  e^{-\frac{2 \lambda}{\mu}  \left(1 - e^{-\mu x}\right)}  e^{\frac{2 \lambda}{\mu}  \left(1 - e^{-\mu \excNL{}(x)}\right)}  e^{-2 \lambda \excNL{}(x)},\\
\overline{g_{\N}}(x) &= \left(1 - e^{-\mu x}\right)  e^{-\frac{2 \lambda}{\mu}  \left(1 - e^{-\mu \excLN{}(x)}\right)}  e^{\frac{2 \lambda}{\mu}  \left(1 - e^{-\mu x}\right)}  e^{-2 \lambda x}.
\end{align}
\end{remark}

\subsection{Association Probability}
Let  $E_s$ denote the event that the typical user is served by a BS of a particular type $s$. The association probability $A_s$ denotes the probability of this event and is given in the following Lemma.
\begin{lemma}
The LoS and NLoS association probabilities of the typical user in a 1D mmWave cellular network are  
\begin{align}
A_\L=\prob{E_\L}=G_\L(0)=&2\lambda \int_0^{\infty} \probserv{\L}(x) \dd x,\\
A_\N=\prob{E_\N}=G_\N(0)=&2\lambda \int_0^{\infty} \probserv{\N}(x)  \dd x
\end{align}
with $\probserv{\L}(x)$ and $\probserv{\N}(x)$ given in \eqref{eq:gLx} and \eqref{eq:gNx}.
\end{lemma}
\begin{cor}

For the special case \LAOP model, the LoS association probability  of the typical user in a 1D mmWave cellular network is given as 
\begin{align}
A_\L=\frac{{\frac\lambda\mu}{\left(2+\frac\lambda\mu\right)}}{\left(1+\lambda/\mu\right)^2}\label{eq:1DAL}
\end{align}
whereas under the independent blocking assumption, it is given as
\begin{align}
\overline{A_\L}=1-\expU{-2\lambda/\mu}.
\end{align}
\end{cor}
Note that these two expressions are very different from each other which confirms the importance of including blockage correlation in the analysis. 
Note that \eqref{eq:1DAL} can be written as
\begin{align}
A_\L=1-\frac{1}{\left(1+\lambda/\mu\right)^2}\le \overline{A_\L}
\end{align}
which indicates that independent blocking assumption overestimates the LoS association probability. 
We will present a quantitative comparison in the numerical section. 

\subsection{Interference and SINR Coverage}
The sum interference $I$ at the typical user depends on the distance and the blocking state of the serving BS, and the locations $\qr$ and $\ql$ of the closest blockages. Therefore, we  derive the  LT of the interference conditioned on these parameters first.  
\begin{lemma}\label{lem:LTI1D}
Conditioned on the event $E_{s,\x}$, and the locations $\qr$ and $\ql$ of the closest blockages, the LT of the interference at the typical user in a 1D mmWave cellular network is given as
\begin{align}
   \mathcal{L}_{I| E_{\L,\x},\ql,\qr}(s)&= \exp\left( -\lambda \left( \int_{\qr}^{\infty} 
   \frac{1}
   {
   1 + {(s \ell_{\N}(y))}^{-1}
   } 
   \dd y + 
   \int_{\excNL{}(x)}^{\infty}
   \frac{\indside{y>\ql}}{
   1 + {(s \ell_{\N}(y))}^{-1}
   } 
    \dd y\nonumber \right. \right.\nonumber\\
    & \left. \left. \qquad + \int_{x}^{\qr} 
    \frac{1}{
   1 + {(s \ell_{\L}(y))}^{-1}
   } 
     \dd y + \int_{x}^{\infty}
      \frac{\indside{y<\ql}}{
   	1 + {(s \ell_{\L}(y))}^{-1}
   	} 
    \dd y
      \right)
       \right),\label{eq:1DLTIL}\\
   \mathcal{L}_{I| E_{\N,\x},\ql,\qr}(s)&= 
   \exp \left( -\lambda  \left( \int_{x}^{\infty}
      \frac{1}{
   1 + {(s \ell_{\N}(y))}^{-1}}
   \dd y
   + 
   \int_{x}^{\infty} 
      \frac{\indside{y>\ql}\dd y}{
   1 + {(s \ell_{\N}(y))}^{-1}}
    \dd y\nonumber \right. \right.\nonumber\\
    & \left. \left. \quad \qquad + 
    \int_{\excLN{}(x)}^{\infty} 
    \frac{\indside{y<\ql}}
    {1+(s \ell_{\L}(y))^{-1}}
     \right) \right).\label{eq:1DLTIN}
\end{align}
\end{lemma} 
\begin{IEEEproof}
See Appendix \ref{appendix:proof-interference} for the proof.
\end{IEEEproof}

Equipped with the distribution of the distance to the serving BS and LT of the interference, we now give the expressions for the SINR coverage   in the following Theorem.
\begin{theorem}\label{thm:Pc1D}
The SINR coverage  of the typical user in  a 1D mmWave network is given by
\begin{align}
\Pc(\SThres) &=
2 \lambda \int_0^{\infty} \int_{x}^{\infty} \int_0^{\infty}  
	\exp\left( - \frac{\SThres  \noise}{ \ell_{\L}(x)} 
	\right)
	\kappa_\L(\SThres,x,\ql,\qr)\ 
	e^{-\mu\ql-\mu\qr}\ 
	g_\L(x \mid \ql,\qr) \ 
	 \dd \ql  \dd \qr  
	 \dd x  \nonumber\\
&+2 \lambda \int_0^{\infty} \int_{0}^{x} \int_0^{\infty}  
	\exp\left( - \frac{\SThres  \noise}{ \ell_{\N}(x)} 
	\right)
	\kappa_\N(\SThres,x,\ql,\qr)
	e^{-\mu\ql-\mu\qr}
	g_\N(x \mid \ql,\qr) 
	 \dd \ql  \dd \qr  
	 \dd x 
\end{align}
where $g_\L(\cdot)$ and $g_\N(\cdot)$ are given in \eqref{eq:cond_gL} and \eqref{eq:cond_gN} respectively, and
\begin{align*}
 \kappa_\L(\SThres,x,\ql,\qr)=
   \mathcal{L}_{I| E_{\L,\x},\ql,\qr}\left( \frac{\SThres}{ \ell_{\L}(x) } \right)
   &= \exp\left( -\lambda \left( \int_{\qr}^{\infty} 
   \frac{\dd y}
   {
   1 + {{(\SThres)}^{-1}\frac{\ell_{\L}(x)}{\ell_{\N}(y)}
   }} 
    \nonumber \right. \right.\\
    &\hspace{-2in} \left. \left.
    + 
   \int_{\excNL{}(x)}^{\infty}
   \frac{\indside{y>\ql} \dd y}{
   1 + {{(\SThres)}^{-1}\frac{\ell_{\L}(x)}{\ell_{\N}(y)}}
   } 
    + \int_{x}^{\qr} 
    \frac{
    \dd y}{
   1 + {{(\SThres)}^{-1}
   \frac{\ell_{\L}(x)}{\ell_{\L}(y)}}
   } 
     + \int_{x}^{\infty}
      \frac{\indside{y<\ql}\dd y}{
   	1 + {{(\SThres)}^{-1}
		\frac{\ell_{\L}(x)}{\ell_{\L}(y)}}
   	} 
      \right)
       \right),\\\\
       \kappa_\N(\SThres,x,\ql,\qr)=
        \mathcal{L}_{I| E_{\N,\x},\ql,\qr}\left( \frac{\SThres}{ \ell_{\N}(x)} \right)&= 
   \exp \left( -\lambda  \left( \int_{x}^{\infty}
      \frac{\dd y}{
   1 + {
   {(\SThres)}^{-1}
   \frac{\ell_{\N}(x)}{\ell_{\N}(y)}}}
   + 
   \int_{x}^{\infty} 
      \frac{ \indside{y>\ql}\dd y}{
   1 + 
   {(\SThres)}^{-1}
   \frac{\ell_{\N}(x)}{\ell_{\N}(y)}
   }
   \nonumber \right. \right.\\
    & \left. \left. \quad \qquad + 
    \int_{\excLN{}(x)}^{\infty} 
    \frac{\indside{y<\ql} \dd y}
    {1+
     {(\SThres)}^{-1}
   \frac{\ell_{\N}(x)}{\ell_{\L}(y)}
    }
    \right) \right).
\end{align*}
\end{theorem}
\begin{IEEEproof}
See Appendix \ref{app:Pc1DProof}.
\end{IEEEproof}

Theorem \ref{thm:Pc1D} shows that the expressions obtained while considering the blocking correlation are significantly different than the expressions under the independent blocking assumption reported in the past literature. In the next section, we numerically evaluate the derived expressions of the SINR coverage to provide a quantitative comparison with respect to the independent blocking assumption.

\subsection{Numerical Results} 
We now present numerical results to quantify the impact of correlation among blocking events  and provide valuable insights. We have considered BPLP path-loss model. The LOS path loss coefficient $\alpha_\L$ is $2.2$, path loss at unit distance $\cgain{\L}=10^{-6}$. The NLOS path loss coefficient $\alpha_\N$ is $3.6$, path loss at unit distance $\cgain{\N}=10^{-7}$. $\cgain{s}$ includes the transmit power of 1 W. The noise power is $-174$ dBm/Hz. Bandwidth is 1GHz. The BS density $(\lambda)$ taken is 10 BS/$\mathrm{km}$, the blockage density is $\beta=\mu$ is 0.007/m.

\textbf{Impact of blocking correlation on the accuracy of analysis:} Fig. \ref{fig:sinrcov} shows the coverage probability  for various SINR thresholds. We have also included the coverage probability under independent blocking assumption. We can see that for low to moderate thresholds, the two curves are different from each other. The result demonstrates the importance of considering the blocking correlation in the analysis of coverage probability. This shows that for such thresholds, correlation plays a significant role on the coverage probability.
\begin{figure}[ht!]
    \centering
    \includegraphics[width = .5\textwidth]{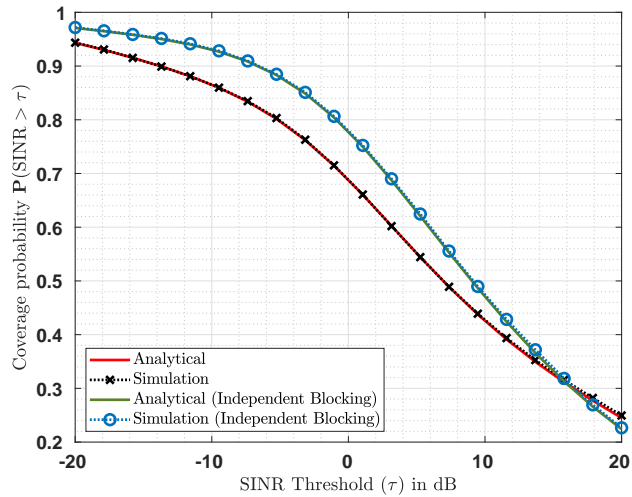}
    \caption{Variation of the SINR coverage of a 1D mmWave cellular network with the SINR threshold ($\SThres$). The coverage probability under independent blocking assumption is also shown. The result demonstrates the importance of considering the blocking correlation in the analysis of coverage probability.
    }
    \vspace{-.1in}
    \label{fig:sinrcov}
\end{figure}

\begin{figure}[ht!]
    \centering
    \includegraphics[width = .5\textwidth]{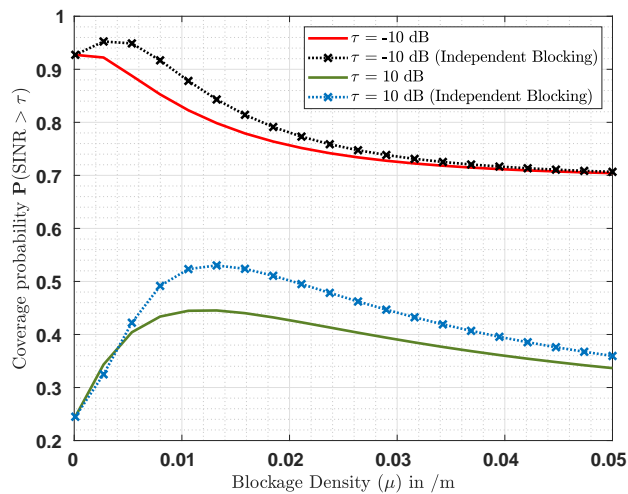}
    \caption{Variation of the SINR coverage  of a 1D mmWave cellular network versus blockage density ($\mu$) for two values of SINR threshold. It is crucial to include blocking correlation in the analysis of mmWave networks.
    }
    \vspace{-.2in}
    \label{fig:sinrvsmu1D}
\end{figure}

\textbf{Impact of blockage density:} Fig. \ref{fig:sinrvsmu1D} shows how correlation impacts the accuracy of the predictions for various values of blockage density $\mu$. We can see that for independent blocking assumption can significantly overestimate the  coverage probability for a wide range of blockage densities. We can also observe that the presence of blockages does not always have a negative impact on the coverage. Up to a certain density, blockages can help reduce the LoS interference, resulting in an improvement in the SINR coverage. The value of optimal blockage density is also incorrectly predicted by the independent blocking assumption. These results show that it is crucial to include blocking correlation in the analysis of 1D mmWave networks.

We can observe that the coverage is significantly different than the coverage obtained under independent blocking assumption. As shown by the 1D analysis, the blockage correlation can significantly change the performance of network which motivates it to further evaluate it for a 2D network. From the above analysis, it is also evident that the exact analysis of blockage correlation requires conditioning on the locations of closest blockages in all directions. In 2D, this becomes difficult as there are multiple directions which prevents the exact analysis. Therefore, we present an approximate but fairly accurate analysis for the 2D case.

 \section{2D Cellular Network}\label{sec:2D}
We now analyze  a 2D  mmWave cellular network \ie the deployment space $\sE=\mathbb{R}^2$. 
Since the blocking state of BSs in 2D may depend on the locations of many blockages, the exact analysis is not possible for the 2D case. Hence, we will take a different approach  based on  the correlation between blocking events of the links \cite{gupta2017macrodiversity} to analyze the 2D case.

\subsection{Correlation between Blocking Events of Links}
Let us consider two BSs located at distance $r_1$ and $r_2$ from the user with angle $\theta\in[0,\pi]$ between the two links. Since a blockage can block both links together, the blocking events of these two links are correlated. From \cite{gupta2017macrodiversity}, the joint LoS probability of the two links (\ie the probability that both the links are LoS) is
\begin{align}
    \label{eqn:joint_prob}
    \probln{\L,\L}(r_{\mathrm{l}},r_{\mathrm{2}},\theta) 
    & =\exp(-\mu \mathcal{N}(r_{\mathrm{1}},r_{\mathrm{2}},\theta))\\
\text{with }\quad \mathcal{N}(r_{\mathrm{l}},r_{\mathrm{2}},\theta) &= \int_0^{\infty} \int_0^{\pi} A(r_{\mathrm{1}},r_{\mathrm{2}},\theta,l,\delta ) F_{\mathrm{L}}(dl) F_{{\Delta}}(d\delta). \label{eq:Ndef}
\end{align}
Here, $A(r_{\mathrm{1}},r_{\mathrm{2}},\theta,l,\delta )$ is the area of the union of two parallelograms with one common edge having length $l$ and angle $\delta$ and the other edge of lengths $r_1$ and $r_2$ respectively with angle $\theta$ in between. This area represents the region in which the center of the blockage should not lie to ensure both the links are LoS.
It is given as
\begin{align}
A(r_{\mathrm{1}},r_{\mathrm{2}},&\theta,l,\delta ) = l r_{\mathrm{1}} \sin(|\delta-\theta|) + lr_{\mathrm{2}} \sin(\delta) -\indside{\delta>\theta} \nonumber \\ 
&  \frac{l^2\sin(\delta)\sin(\delta-\theta)}{2 \sin(\theta)} \times \left[ 1-\left(1 - \min\left( 1,\frac{r_{\mathrm{1}} \sin(\theta)}{l \sin(\delta)}\frac{r_{\mathrm{2}} \sin(\theta)}{l \sin(\delta-\theta)} \right)\right)^2 \right].\label{eq:ADef}
\end{align}

Substituting \eqref{eq:ADef} in \eqref{eq:Ndef}, $\mathcal{N}(r_{\mathrm{l}},r_{\mathrm{2}},\theta)$ can also be  written as
\begin{align}
&\mathcal{N}(r_{\mathrm{l}},r_{\mathrm{2}},\theta)=\beta/\mu (r_1+ r_2)-\F'(r_1,r_2,\theta)\end{align}
\begin{align}
&\text{where } 
\F'(r_1,r_2,\theta) =\mathbb{E}_{L,\Delta}\left[\indside{\Delta>\theta}  \frac{L^2\sin(\Delta)\sin(\Delta-\theta)}{2 \sin(\theta)}  \nonumber\right. \\
&\qquad \qquad \qquad \times \left.\left[
	 1-\left(1 - \min
	 	\left( 1,\frac{r_{1} \sin(\theta)}{L\sin(\Delta)}\frac{r_{\mathrm{2}} \sin(\theta)}{L \sin(\Delta-\theta)} \right)
		\right)^2 
	 \right]\right].\label{eq:Fdashdef}
\end{align}
Here, the blocking parameter $\beta = {2\mu\expect{L}}/{\pi} $ where $\expect{L}$ is average length of a blockage \cite{bai2014coverage}. 
Since $F'\ge0$,  $\mathcal{N}$ is upper bounded by $\beta (r_1+r_2)$. Thus, from \eqref{eqn:joint_prob}, 
\begin{align}
    \probln{\L,\L}(r_{\mathrm{l}},r_{\mathrm{2}},\theta) 
    & >\probln{\L}(r_{\mathrm{l}}) \probln{\L}(r_{\mathrm{2}})= \exp(-\beta r_1-\beta r_2) .\label{eqn:joint_prob_lb}
\end{align}
Note that the the RHS is the LoS probability under the independent blocking assumption which is the case when blockages have  very small lengths. This is due to the fact that a small blockage blocks only one link at a time. When the blockage length is large, a single blockage may block multiple links simultaneously resulting in a higher correlation. We now observe the impact of scaling of the blockage length.
\begin{lemma}\label{lemma:blockcorr_monotonic}
If we increase blockages' length while keeping $\beta$ constant, the value of $F'$ given in \eqref{eq:Fdashdef} monotonically increases and hence, $\probln{\L,\L}(r_{\mathrm{l}},r_{\mathrm{2}},\theta) $ increases. This indicates an increase in the blockage (or LOS) correlation.
\end{lemma}
\begin{IEEEproof}
The proof is similar to Lemma 1 of \cite{gupta2017macrodiversity}.
\end{IEEEproof}
Lemma \ref{lemma:blockcorr_monotonic} provides us with the following upper bound  \cite{gupta2017macrodiversity} on the joint LoS probability
\begin{equation}
\label{eqn:joint_prob_ub}
    \probln{\L,\L}(r_1,r_2,\theta) \leq \exp\Big(-\beta r_1 - \beta r_2 +\frac{\min(r_1,r_2)}{2}\beta (1+\cos(\theta))\Big).
\end{equation}
The upper bound represents the value corresponding to the maximally correlated case \ie the case where blockages tend to be of infinite length. 

The event that a link is LoS, eliminates the chances of the presence of any blockage on the link. Therefore, it increases the LoS probability of nearby links. That is the reason why  the independent blocking case represents the lower bound on the joint LoS probability whereas the upper bound is corresponding to the case with infinite  blockage length.

Similarly, the joint probability that the first link is LoS and the other is NLoS is given as
\begin{align}
    \probln{\L,\N}(r_{\mathrm{l}},r_{\mathrm{2}},\theta) & =\exp(-\beta r_{\mathrm{1}})-\exp(-\mu \mathcal{N}(r_{\mathrm{1}},r_{\mathrm{2}},\theta)).
    \end{align}
Also, the joint probability that the first link is NLoS and other is LoS is  
\begin{align}
    \probln{\N,\L}(r_{\mathrm{1}},r_{\mathrm{2}},\theta) & = \exp(-\beta r_{\mathrm{2}})-\exp(-\mu \mathcal{N}(r_{\mathrm{1}},r_{\mathrm{2}},\theta)) ,
\end{align}
and the probability that both links are NLoS is
\begin{align}
    \probln{\N,\N}(r_{\mathrm{1}},r_{\mathrm{2}},\theta) & = 1-\exp(-\beta r_{\mathrm{1}})-\exp(-\beta r_{\mathrm{2}}) +  \exp(-\mu \mathcal{N}(r_{\mathrm{1}},r_{\mathrm{2}},\theta)).
\end{align}

Further, the joint  LoS probability of three and more links can be computed, however the expressions become very complex. Since the blocking state of BSs may depend on the locations of many blockages, the exact analysis is not possible for the 2D case. Hence, in this work, we  consider only the first order correlation in LoS probabilities for simplicity. In particular, we consider the correlation between $\s0$ and $\s{i}\ \forall i$ \ie blocking in the serving link and each interfering link and  ignore the correlation between interfering links (\ie between $\s{i}$ and $\s{j}$, $i\ne j;i,j\ne0$). However, the analysis can be extended to consider higher order of correlation. In Section \ref{chap:num_results}, we will show that the effect of this approximation on the accuracy of results is minimal  and demonstrate that the first order assumption is indeed able to provide sufficient accuracy.

\subsection{Distance Distribution of the Serving LoS and NLoS BS}
\label{chap:d_serving}

For the 2D case, the probability $G_{\L}(r)$  that the serving BS is  located outside the 2D ball $\Ball(\origin,r)$ and is LoS, is given as
\begin{align}
        &G_{\L}(r) =\expects{\Phi}{\sum_{\X \in \Phi}
        \indside{E_{\L,\X}} \indside{\X\notin\Ball(\origin,r)}}
        = 2\pi\lambda  \int_{r}^\infty 
       \probserv{\L}(x) 
        x \dd x\label{eq:GLr1},\\
&\text{with } \probserv{\L}(x) \!\isdefined \palmprobx{E_{\L,\x}}{\x}\\
& \approx \!\probln{\L}({x}) \!
\expS{\!-\lambda\int_0^{2\pi} \!\!\!\int_0^{{x}}
\!\frac{\probln{\L,\L}({x},t,\theta)}{\probln{\L}({x})}t\dd t \dd \theta
\!-\lambda\int_0^{2\pi} \!\!\!\int_0^{\excNL{}({x})}
\!
\frac{\probln{\L,\N}({x},t,\theta)}
{\probln{\L}({x})}
t\dd t \dd\theta\!}\label{eq:glx}.
\end{align}

Similarly, we get the probability $G_{\N}(r) $ that the typical user is served by an NLoS  BS   located outside the ball $\Ball(\origin,r)$ as
\begin{align}
    &G_{\N}(r) = 2\pi \lambda \int_r^{\infty} \probserv{\N}(x) x \dd x,\\
&\text{with }
\probserv{\N}(x) \! \approx
\! \probln{\N}({x}) \!
\expS{\!-\lambda\!\int_0^{2\pi} \!\!
\int_0^{\excLN{}({x})}\!\!
\frac{\probln{\N,\L}({x},t,\theta)}{\probln{\N}({x})}t\dd t \dd \theta
\!-\!\lambda\int_0^{2\pi}\!\! \int_0^{{x}}
\frac{\probln{\N,\N}({x},t,\theta)}
{\probln{\N}({x})}
t\dd t \dd\theta}\label{eq:gNx2D}.
\end{align}

Further, the association probability $A_s$ \ie the probability of the event  that the typical user is served by a $s$ type BS is given in the following Lemma.
\begin{lemma}
The LoS and NLoS association probabilities of the typical user in a 2D mmWave cellular network are  given as
\begin{align}
A_\L=\prob{E_\L}=G_\L(0)=&2\pi \lambda \int_0^{\infty} \probserv{\L}(x) x \dd x,\label{eq:LOS_assoc_prob}\\
A_\N=\prob{E_\N}=G_\N(0)=&2\pi \lambda \int_0^{\infty} \probserv{\N}(x) x \dd x\label{eq:NLOS_assoc_prob}
\end{align}
where $\probserv{\L}$ and $\probserv{\N}$ are given in \eqref{eq:glx} and \eqref{eq:gNx2D}.
\end{lemma}

Equipped with the expressions of $G_s(r)$ and $A_s$, we now give the PDF of the distance of the serving BS conditioned on the event $E_s$. Note that the angular distribution of the location of the serving BS is uniform.
\begin{lemma}
For the typical user in a 2D mmWave cellular network, the PDF of its serving BS distance conditioned on $E_\L$ is
\begin{align}
f_{X_0|E_\L}(x)&=\frac1{A_\L}2\pi\lambda \probserv{\L}(x)x, \ \ (x\ge0).
\end{align}
Similarly, the PDF of its serving BS distance conditioned on $E_\N$ is
\begin{align}
f_{X_0|E_\N}(x)&=\frac1{A_\N}2\pi\lambda \probserv{\N}(x)x, \ \ (x\ge0).
\end{align}
\end{lemma}

\subsection{Bounds}
We now present some simplified  expressions acting as the upper and lower bound for $g_s(x)$. 

\begin{lemma}
\label{lemma:glxgeneralbounds}
In a 2D mmWave network, the function $\probserv{\L}(x)$ as given in \eqref{eq:glx} decreases when the blocking correlation increases. Further, it can be  bounded as
\begin{align}
  	\probserv{\L}(x)			&	\le \overline{\probserv{\L}}(x)
  	\!= \!\expS{-\beta x
	-\!\lambda\pi\excNL^2{}\!({x})
	-\! \frac{2\pi \lambda}{\beta^2} \!
	(F_1(\excNL{(x)})-F_1(x))}\label{eq:G_Ind_Gen}\\
\probserv{\L}(x) &\geq\underline{g_\L}(x)  =e^{-\beta x} \expS{-\lambda\pi\excNL^2{}({x})-4 \lambda \int_{\excNL{}({x})}^x \int_0^{\pi/2} t \expU{-\beta t \sin^2 \theta } \dd \theta \dd t}\label{eq:G_Inf_p1_Gen} \\
&\hspace{.55in}=
  \expS{-\beta x -\lambda\pi\excNL^2{}({x}) -
  \frac13 \frac{2\pi\lambda}\beta [F(x,\beta)-F(\excNL{}(x),\beta)]
  },\label{eq:G_Inf_Gen}
\end{align}
 $$\text{with } F_1(z)= (1+\beta z)\expU{-\beta z},\ \text{and } \ F (r,\beta)= r e^{-\beta r/2}\left(
  \beta r\IB_0(\beta r/2)+(\beta r+2)\IB_1(\beta r/2)
  \right).$$
Here,  $\overline{\probserv{\L}}(x)$ corresponds to the value of $\probserv{\L}(x)$ under the independent blockage assumption and $\underline{g_\L}(x)$ corresponds to the  value  of $\probserv{\L}(x)$ for the maximally correlated case.
\end{lemma}
\begin{IEEEproof}
See Appendix \ref{app:glxgeneralbounds}.
\end{IEEEproof}

\begin{remark}
Lemma \ref{lemma:glxgeneralbounds} shows that in a 2D mmWave network, the value of $g_{\L}(r)$ when blocking correlation is considered, is less than the respective value obtained under independent blocking assumption and greater than the value obtained under the maximally correlated case.
\end{remark}

\begin{remark}
From \eqref{eq:GLr1} and \eqref{eq:LOS_assoc_prob}, it can be seen that  the CCDF of  LoS serving BS distance $G_\L(r)$ and the LOS association probability $A_\L$  in a 2D mmWave cellular network decrease when blockage correlation increases. Therefore, independent blocking assumption and maximally correlated case act as upper and lower bounds respectively for these metrics also. 
\end{remark}

Similarly, the bounds can be derived for $g_\N$ as given in the following Lemma.
\begin{lemma}
\label{lemma:gnxgeneralbounds}
In a 2D mmWave network, the function $\probserv{\N}(x)$ 
increases when the blocking correlation increases. Further, it can be  bounded as
\begin{align}
  	\probserv{\N}(x)			&	\le \overline{\probserv{\N}}(x)
  	\!= (1-\expU{-\beta x})\expS{\!
	-\!\lambda\pi x^2\!\!
	+\! \frac{2\pi \lambda}{\beta^2} \!
	\left(F_1(\excLN{(x)})-F_1(x)\right)\!},\nonumber
			\\
\probserv{\N}(x) &\geq\underline{g_\N}(x)  = (1-e^{-\beta x}) 
\expS{-\lambda\pi x^2+
\frac{\lambda 2 \pi}{\beta^2}(F_1(\excLN{(x)})-F_1(x))
 \frac{1-e^{-\beta x/2}\IB_0(\beta r/2)}{1-e^{-\beta x}} 
  }
  \nonumber
\end{align}
\end{lemma}

\subsection{SINR Coverage}
\label{chap:SINR_performance}
Using the law of total probability, the SINR coverage $\Pc(\SThres)$ as defined in \eqref{eq:pcdef} is given as
\begin{align*}
 \Pc(\SThres)=A_\L\prob{\SINR>\SThres|E_\L}+A_\N\prob{\SINR>\SThres|E_\N}.
 \end{align*}
 Here, $\Pci{s}(\SThres)\stackrel{\Delta}{=}\prob{\SINR>\SThres|E_s}$ denotes the  SINR coverage  conditioned on the association with a $s$-type BS. Using conditioning on the serving BS distance $X_0$, we get
\begin{align}
    \Pci{s}(\SThres) &
     = \int_0^{\infty} 
    \prob{\SINR>\SThres|X_0=r,E_\u}
    f_{X_0|E_\u}(r)\dd r\nonumber\\
     & 
	= \int_0^{\infty}
    \prob{\frac{H_0 \ell_{\s{}}(r)}{(I+\noise)}>\SThres|X_0=r,E_\u}
    f_{X_0|E_\u}(r)\dd r\nonumber\\
     & \stackrel{(a)}{=}\int_0^{\infty} \!\!
    \expects{\Phi}{\probs{H_0}{H_0>
    \frac{\SThres{(I+\noise)}}
    {\ell_{\s{}}(r)}
    }|X_0=r,E_\u}
   f_{X_0|E_\u}(r)\dd r\nonumber\\
     &
     \stackrel{(b)}= \int_0^{\infty}
    \expects{\Phi}{\expS{-
    \fracS{\SThres{(I+\noise)}}
    {\ell_{\u{}}(r)}
    }|X_0=r,E_\u}
    f_{X_0|E_\u}(r)\dd r\nonumber\\
     &
    \stackrel{(c)}=  \frac1{A_\u}\int_0^{\infty} 2 \pi \lambda  e^{-
    \SThres{\noise}/\ell_{\u}(r)
    }
    \laplace{I|X_0=r,E_\u}{
    \left(
    \SThres/\ell_{\u{}}(r)
    \right)}
    \probserv{\u}(r)r \dd r\label{eq:pc1}
\end{align}
where $( a)$ is due to the law of conditional probability, $(b)$ uses $H_0\sim\mathsf{Exp}(1)$ and $(c)$ is due to definition of the LT. To proceed further, we would need the LT of $I$ conditioned on $X_0$ and $E_s$.

Under the first order correlation assumption, conditioned on the serving BS type, type of each BS is independent of other BSs' types. Therefore, $\Phi_\L$  and $\Phi_\N$ are mutually independent and hence, $I_\L $ and $I_\N$ are also. Therefore,
\begin{align}
\laplace{I|X_0=r,E_s}(\nu)=\laplace{I_\L|X_0=r,E_\u}(\nu)\laplace{I_\N|X_0=r,E_\u}(\nu).\label{eq:laplaceproduct}
\end{align}
The LT of $I_\v$ can be derived using PGFL of PPP and is given in the following theorem.
\begin{theorem}\label{thm:laplaceI}
Conditioned on the serving BS's type $\u$ and its distance $X_0$ from the user, the LT of the   interference $I_v$ caused by $\v$-type interferers in a 2D mmWave network is approximately given as
\begin{align}
 \laplace{I_\v|X_0=r,E_\u}(\nu) &\approx \expS{-\frac{
  %2 \pi 
  \lambda}{\probln{\u}(r)} 
  \int_{\exc{\v|\u}(r)}^{\infty}
  \frac{\nu 
  \ell_\v(t)
  			}
  {1+\nu 
   \ell_\v(t)
  } 
\int_0^{2\pi}
{\probln{\u,\v}(r,t,\theta)}\; \dd\theta\; t\; \dd t }\label{eq:LTI2D}.
\end{align}
\end{theorem}
\begin{proof}
See Appendix \ref{app:thm:laplaceI}.
\end{proof}

\begin{remark}
It can be shown that under LoS association, the interference from LoS BSs increases statistically while NLOS interference decreases when blocking correlation is present, in comparison to the independent blocking assumption. An opposite behavior is seen under NLOS association.
\end{remark}
Now, substituting \eqref{eq:LTI2D} in \eqref{eq:pc1}, we get the SINR coverage   as given in the following theorem.

\begin{theorem}
\label{theorem: SINR}
The SINR coverage  $\Pc(\SThres)$ of the typical user in a 2D mmWave cellular network in the presence of blockage correlation is approximately given as
\begin{align}
 \Pc(\SThres)\! \approx \! \!\sum_{\u\in\{\L,\N\}}\int_0^{\infty} 2 \pi \lambda  e^{-\SThres{\noise}/\ell_\u(r)
    }
    \Lfunc_{\L|\u}(r,\SThres)\,\Lfunc_{\N|\u}(r,\SThres)
   \, \probserv{\u}(r)r \dd r.\label{eq:PcFinalEq}
    \end{align}
    with $g(s)$ given in \eqref{eq:glx} and \eqref{eq:gNx2D}, and 
    \begin{align}
 \Lfunc_{\v|\u}(r,\tau)\stackrel{\Delta}{=} \laplace{I_\v|X_0=r,E_\u}\left(
    \frac{\SThres{}}{\ell_s(r)}
    \right)=\expS{-
  \int_{\exc{\v|\u}(r)}^{\infty}
  \frac{
  \lambda
  			}
  {1+
  \SThres^{-1} 
  \frac{\ell_\u(r)}{\ell_\v(t)}
  } 
\int_0^{2\pi}
\frac{\probln{\u,\v}(r,t,\theta)}
{\probln{\u}(r)} 
 \dd\theta t \dd t }\label{eq:kvs}.
\end{align}
\end{theorem}

\subsection{Rate Coverage}
If the typical user is allotted a bandwidth of $B_\ue$, the instantaneous rate  is 
\begin{align}
	\Rate = B_\ue \log(1+\SINR) \label{eq: def_rate}.
\end{align}
Since tagged BS is associated with multiple users, the BS need to divide the available time-frequency resources among all  connected users. Therefore, the value of $B_\ue$ depends on the number of users in the tagged cell which is a RV. It is shown that taking average number of users does not affect the result significantly \cite{bai2013Globesip}, \cite{larsson2014massive}. For bandwidth allocation to users, we consider the following two cases.
 
\subsubsection{Equal allocation}
In this case, each BS divides the available bandwidth equally among all UEs.  Therefore, the bandwidth allocated to the typical UE of the tagged BS is given as
\begin{align}
	B_\ue = \frac{B}{N_\ue}.
\end{align}
Here, $N_\ue$ denotes the mean number of users connected to the tagged BS and is given as \cite{SinKulGhoAnd2015,singh2013offloading}:
\begin{align}
	N_\ue = 1+1.28\frac{\lambda_u}{\lambda}
\end{align}
From  \eqref{eq: def_rate}, we can derive the rate coverage which is given in the following Lemma.
\begin{lemma}\label{lem:rc2D}
The rate coverage of the typical user in a 2D mmWave cellular network under equal bandwidth allocation to all users is given as
\begin{align}
	\Rc(\rho) &= \prob{B_u\log(1+\SINR) > \rho} = \prob{\frac{B}{N_u}\log(1+\SINR) > \rho}\nonumber \\
	&= \prob{\SINR > 2^{\frac{\rho N_u}{B}}-1} 
	= \Pc(2^{\frac{\rho N_u}{B}}-1)
\end{align}
where $\Pc(\cdot)$ is given in Theorem \ref{theorem: SINR}.
\end{lemma}

\subsubsection{LoS allocations only}
Since NLoS users can have very poor SINR, an equal allocation to NLoS and LoS users may not be a good idea. Here, we consider that each BS divides available bandwidth  among only LoS UEs and NLoS UEs are not served. Further, the average number of LoS UEs $N_{\L u}$ connected to the tagged BS and is given by \cite{singh2013offloading}
\begin{align}
	N_{\L \ue} = 1+1.28\frac{\lambda_\ue A_\L}{\lambda}.
\end{align}
Hence, the bandwidth allocated to any LoS UE of the tagged BS is given as:
\begin{align}
	B_{\L u} = \frac{B}{N_{\L u}}.
\end{align}
Now, using the law of total probability, the rate coverage is given as
\begin{align}
	\Rc(\rho) &= \prob{\Rate > \rho} = A_\L \prob{\Rate > \rho|E_\L}+  A_\N \prob{\Rate > \rho|E_\N} \nonumber\\
	&= A_\L\prob{\frac{B}{N_{\L \ue}}\log(1+\SINR) > \rho|E_\L} = A_\L\prob{\SINR > 2^{\frac{\rho N_{\L \ue}}{B}}-1|E_\L} = A_\L\Pci{\L}\left(2^{\frac{\rho N_{\L \ue}}{B}}-1\right) \label{eq: onlyLoSband}
\end{align}
where $\Pci{\L}$ is given in \eqref{eq:pc1}.

\subsection{Special Case: \LAOP model }
\label{sec:2DSpecialCase}
In the last subsection, we derived the SINR coverage for the general path-loss model. In this subsection, we consider the special case \LAOP model where $\ell_\N(r)$ is very low  for all values of  $r$ resulting in the complete outage for all NLoS links. This special case helps us  simplify expressions further.

Note that $\excNL{}(r)=0$ for this case. Therefore, from \eqref{eq:glx}, 
\begin{align}
\probserv{\L}(x) & = \probln{\L}({x}) 
\expS{-\lambda\int_0^{2\pi} \int_0^{{x}}\frac{\probln{\L,\L}({x},t,\theta)}{\probln{\L}({x})}t\dd t \dd \theta}.\label{eq:probserv:LOSOnly}
\end{align}
Now, 
\begin{align}
  &\Lfunc_{\L|\L}(r,\tau)\!=\!\expS{-\!\!
  \int_{r}^{\infty}
  \frac{%2 \pi 
  \lambda
  			}
  {1+
  \SThres^{-1}
  \ell_\L(r)/\ell_\L(t)
  } 
\int_0^{2\pi}
\frac{\probln{\L,\L}(r,t,\theta)}
{\probln{\L}(r)} 
 \dd\theta t \dd t },\label{eq:LTI:LOSOnly}
\end{align}
and $\Lfunc_{\N|\L}(r,\SThres)=1$. Therefore, from \eqref{eq:PcFinalEq}, the SINR coverage  is given as
\begin{align}
 \Pc(\SThres) = \int_0^{\infty} 2 \pi \lambda  e^{-
   \SThres{\noise}/
 	\ell_\L(r)
    }\ 
    \Lfunc_{\L|\L}(r,\SThres)\ 
    \probserv{\L}(r)r \dd r.\label{eq:PcFinalEqLOSOnly}
    \end{align}

\noindent The LoS association probability simplifies to 
\begin{align}
&=2\pi \lambda \int_0^{\infty} \probln{\L}({x}) \expS{-\lambda\int_0^{2\pi} \int_0^{{x}}\frac{\probln{\L,\L}({x},t,\theta)}{\probln{\L}({x})}t\dd t \dd \theta} x \dd x \label{eq: A_simp}.
\end{align}
\noindent Since the NLOS links are in complete outage, the BS will not allocate any resources to the NLoS UEs. Therefore, from \eqref{eq: onlyLoSband}, the rate coverage  is given as:
\begin{align}
	R_{\L}(\rho) &= A_\L\Pci{\L}
	\left(2^{\frac{\rho N_{\L u}}{B}}-1\right)
	=\Pc\left(2^{\frac{\rho N_{\L u}}{B}}-1\right).
\end{align}
where is $N_{\L u} = 1+1.28 \fracS{\lambda_u A_\L}{\lambda}$.

We will now present simplified bounds on $\probserv{\L}(r)$, $A_\L$, and $\Lfunc_{\L|\L}(r,\SThres)$  in the following lemmas.

%%%%%%%%%%%%%%%%%%%%%
%%     UPPER BOUNDS
%%%%%%%%%%%%%%%%%%%%%
%%%%%%%%%%%%%%%%%%%%%

\begin{lemma}
\label{theorem:P_ub}
In a 2D mmWave cellular network with LAP model, the function 
$\probserv{\L}(r)$,  the LoS association probability $A_\L$ and the LT of LoS interference ${\Lfunc_{\L|\L}}(r,\tau)$  are upper bounded as
\begin{align}
  	\probserv{\L}(r)			&	\le \overline{\probserv{\L}}(r)
  	= \expS{-\beta r-\frac{2\pi \lambda}{\beta^2} 
			\big(1-(1+\beta r)e^{-\beta r}\big)}\label{eq:G_Ind},\\
	A_\L 						&	\le \overline{A}_\L
	= 1-\exp\left(-\fracS{2\pi\lambda}{\beta^2} \right)\label{eq:A_Ind},\\
	{\Lfunc_{\L|\L}}(r,\tau) 	&	\le \overline{\Lfunc_{\L|\L}}(r,\tau)  
	=\expS{-
				  \int_{r}^{\infty}
				  \frac{
				  2 \pi \lambda \exp(-\beta t)
				  			}
				  {1+
				  \SThres^{-1}
				 \ell_\L(r)/\ell_\L(t)
				  }  t \dd t } \label{eq:L_Ind}.
\end{align}
where  $\overline{\probserv{\L}}(r)$,  $\overline{A}_\L$ and $\overline{\Lfunc_{\L|\L}}(r,\tau)$ correspond respectively to the values of $\probserv{\L}(r)$,  $A_\L$ and  ${\Lfunc_{\L|\L}}(r,\tau)$  under the independent blockage assumption. 
\end{lemma}
\begin{IEEEproof} \eqref{eq:G_Ind} can be obtained by substituting $\excNL{}=0$ in \eqref{eq:G_Ind_Gen}. \eqref{eq:A_Ind} can be obtained by substituting the bound on $g_\L(r)$ as given in \eqref{eq:G_Ind} into \eqref{eq:LOS_assoc_prob}. \eqref{eq:L_Ind} can be derived by substituting  \eqref{eqn:joint_prob_lb} in \eqref{eq:LTI:LOSOnly}.
\end{IEEEproof}

\begin{lemma}
\label{theorem:P_lb}
In a 2D mmWave cellular network with LAP model, the function 
$\probserv{\L}(r)$,  the LoS association probability $A_\L$ and the LT of LoS interference ${\Lfunc_{\L|\L}}(r,\tau)$ are lower bounded as
\begin{align}
    \probserv{\L}(r) &\geq\underline{g_\L}(r)  =e^{-\beta r} \expS{-4 \lambda \int_0^r \int_0^{\pi/2} t \expU{-\beta t \sin^2 \theta } \dd \theta \dd t}\label{eq:G_Inf_p1}\\
 &\hspace{.55in}=
  \expS{-\beta r-
  \frac13 \frac{2\pi\lambda}\beta r e^{-\beta r/2}\left(
  \beta r\IB_0(\beta r/2)+(\beta r+2)\IB_1(\beta r/2)
  \right)
  },\label{eq:G_Inf}
\\
  A_\L &\geq \underline{A_\L}= 2\pi \lambda \int_0^{\infty} \expS{-\beta x
-4\lambda\int_0^{\pi/2} \int_0^{{x}}\expS{-\beta t \sin^2(\theta)} t\dd t \dd \theta} x \dd x, \label{eq:A_Lower}
  \end{align}
  \begin{align}
{\Lfunc_{\L|\L}}(r,\tau) &
	\geq \underline{\Lfunc_{\L|\L}}(r,\tau) =
		\expS{-
		  \int_{r}^{\infty}
		  \frac{
		  \lambda \exp(-\beta t)
		  			}
		  {1+
		  \SThres^{-1}
		  \ell_\L(r)/\ell_\L(t)
		  } 
		 t \dd t \ \cdot
		4 \int_0^{\pi/2}
		\expS{r\beta \cos^2(\theta)}
		 \dd\theta}\nonumber\\
		 &= 
\expS{-
		  \int_{r}^{\infty}
		  \frac{
		  \lambda \exp(-\beta t)
		  			}
		  {1+
		  \SThres^{-1}
		  \ell_\L(r)/\ell_\L(t)
		  } 
		 t \dd t \ \cdot
		2 \pi e^{r\beta/2}\IB_0(r\beta/2)}
		  \label{Eq: L_Lower},
 \end{align}
  which correspond to the values for the maximally correlated case.
\end{lemma}

\begin{IEEEproof} 
To  get \eqref{eq:G_Inf_p1} and \eqref{eq:G_Inf}, we substitute $\excNL{}=0$ in \eqref{eq:G_Inf_p1_Gen} and \eqref{eq:G_Inf_Gen} respectively. \eqref{eq:A_Lower} can be obtained by substituting the values of bounds on $g_\L(r)$ in \eqref{eq:LOS_assoc_prob} and \eqref{Eq: L_Lower} can be obtained by substituting the bound as given in \eqref{eqn:joint_prob_ub} into \eqref{eq:LTI:LOSOnly}.
\end{IEEEproof}

\begin{lemma}\label{lemma:simpleboundgL2D}
In a 2D mmWave cellular network with LAP model, 
further  simplified bounds for $\underline{g_\L}(r)$ are given as 
 \begin{align}
 & \underline{g_\L}(r) \ge\underline{g_1}(r)= \expS{-\beta r 
	-\frac{4\lambda}{m\beta^2(1-\frac{m\pi}{2}) } 
	\left(\frac{m\pi}{2} 
      -\expU{ -\beta r(1-\frac{m\pi}{2})}\right)
      +\frac{4\lambda}{m\beta^2
      } 
      		e^{-\beta r}
      }
      \label{eq:G_lb}\\
\text{and }&\underline{g_\L}(r)\leq\underline{g_2}(r) =
 \expS{-\beta r -\frac{4\lambda}{m^2\beta^2} \left(m\beta r 
     - \frac{2}{\pi}\left(1 -\expU{- \frac{m \beta r \pi}{2}}
     \right)
     \right)}
     \label{eq:G_ub}
\end{align}
where $m=1.38$.
\end{lemma}
\begin{IEEEproof}
See Appendix \ref{appendix:sin1}.
\end{IEEEproof}

Using bounds presented in Lemma \ref{theorem:P_ub} and \ref{theorem:P_lb} in \eqref{eq:PcFinalEqLOSOnly}, we get the following bounds on the SINR coverage. 

\begin{theorem}\label{thm:PcSpecialCase}
In a 2D mmWave cellular network with \LAOP model,
the SINR coverage  $\Pc$ decreases with an increase in the blockage correlation. Further, $\Pc$ can be bounded as
\begin{align}
\underline{\Pc}(\SThres) \le \Pc(\SThres) \le \overline{\Pc}(\SThres)
\end{align}
where  $\overline{\Pc}$ and $\underline{\Pc}$ are  the values of the SINR coverage  under the independent blocking  and  maximally correlated case respectively. These are given as
\begin{align}
 \overline{\Pc} (\SThres)= \int_0^{\infty} 2 \pi \lambda  e^{-
 {\SThres{\noise}r^{\alpha_{\L}}}/
    {\cgain{\L}}
    }\ 
   \overline{ \Lfunc_{\L|\L}}(r,\SThres)
    \overline{\probserv{\L}}(r)r \dd r,\label{eq:PcSCUB}\\
    \underline{\Pc}(\SThres) = \int_0^{\infty} 2 \pi \lambda  e^{-
 {\SThres{\noise}r^{\alpha_{\L}}}/
    {\cgain{\L}}
    }\ 
   \underline{ \Lfunc_{\L|\L}}(r,\SThres)
    \underline{\probserv{\L}}(r)r \dd r.\label{eq:PcSCLB}
    \end{align}
\end{theorem}

The above theorem indicates that independent blocking assumption case serves as an upper bound to real world $\Pc$ while expression for the case with infinite length blockages will give an lower bound on $\Pc$.

\subsection{Numerical Results}
\label{chap:num_results}

We now evaluate the impact of blockage correlation for a 2D deployment. The path-loss model and parameters are the same as 1D case. The BS density $(\lambda)$ taken is 30 BS/$\mathrm{km}^2$, $\beta$ is 0.014/m, and blockage density $(\mu)$ is 220/$\mathrm{km}^2$. The blockages have uniformly distributed length ($L\sim \mathsf{Unif}(0,L_\mathrm{max})$) and orientation ($\Delta\sim \mathsf{Unif}(0,\pi)$). 

\textbf{Distribution of the serving BS distance:} 
We first plot $G_\L(r)$ in Fig. \ref{fig:dist2} which denotes the probability that the serving BS is LoS and farther than distance $r$ from the typical user, along with the derived bounds.  It can be observed that the analysis is close to the exact values (represented by the simulated results) despite the first order correlation assumption. It can be verified that the $\underline{G_\L(r)}$ (maximally correlated case) is the lower bound and $\underline{G_\L(r)}$  (independent case) is the upper bound of the analytical value of $G_\L(r)$. We  observe that the independent blocking assumption does not always gives a good approximation and  incorrectly overestimates the coverage.

\begin{figure}[ht!]
	\begin{center}
		\includegraphics[width=.5\textwidth]{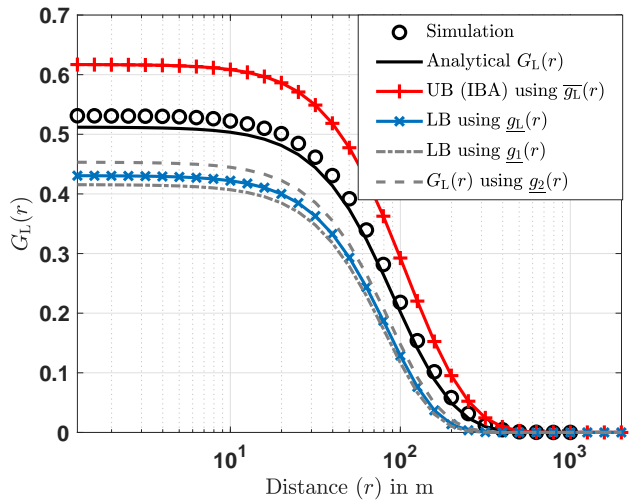}
		\caption[Distribution of distance to the LoS serving base station for 
			$\beta=3\beta_0,\mu = \mu_0$. Distribution of distance to the
			 line of sight serving base station for
			 $\beta=3\beta_0,\mu = \mu_0$]
		{	
			Variation of $G_\L(r)$ (the probability that the serving 
			BS is farther than distance $r$ and LoS from the typical user) 
			with $r$ for a 2D mmWave cellular network in the presence of random blockages. Here, $\beta=.014/m$ and $\mu=220/km^2$. 
			The bounds are evaluated by substituting respective bounds 
			of $\probserv{\L}(r)$ as given in Lemma \ref{lemma:glxgeneralbounds}, in \eqref{eq:GLr1}.  Here, UB and LB refer to upper and lower bounds. IBA refers to independent blocking assumption.
		}
		\label{fig:dist2}
		\end{center}
\end{figure}

\textbf{LoS association probability:} 
Fig. \ref{fig:AP2} shows the impact  of varying blockage correlation on the LoS association probability $A_\L$, along with various derived bounds. Here, we fix the $\beta$ and vary the maximum blockage length $L_\mathrm{max}$. The result shows that the $A_\L$ decreases as blocking correlation increases. For smaller blockages,the independent blocking assumption works well. However, as blocking correlation increases, $A_\L$ deviates significantly from $\overline{A_\L}$ and reaches the lower bound (corresponding to maximally correlated case) as $L_\mathrm{max}\rightarrow \infty$. Therefore, when the blockage length is large, inclusion of the first order correlation in the analysis gives a better accuracy.

\begin{figure}[ht!]
	\begin{center}
				\includegraphics[width=.5\textwidth]
				{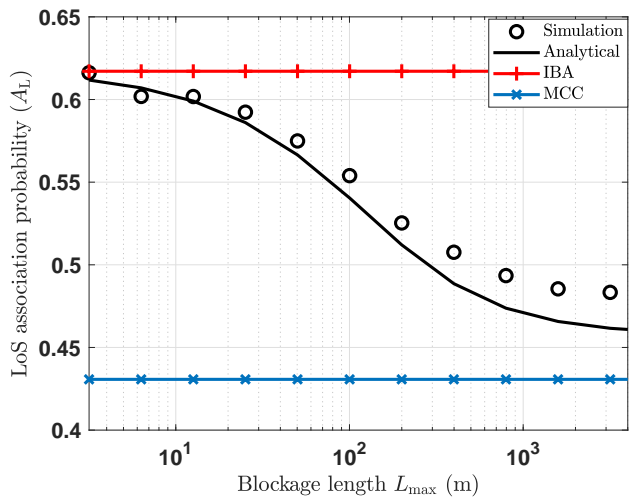} 
				\caption[coverage_probability]
					{Variation of LoS association 
					probability vs maximum blockage length $L_\mathrm{max}$ for a 2D mmWave cellular network in the presence of random blockages. Here, $\beta= 0.014/m$ is fixed. When the blockage length is large, independent blocking assumption no longer remains valid and inclusion of the first order correlation in the analysis gives a better accuracy.}
				\label{fig:AP2}
	\end{center}
\end{figure}

\begin{figure}[ht!]
	\begin{center}			
			\includegraphics[width=.5\textwidth]{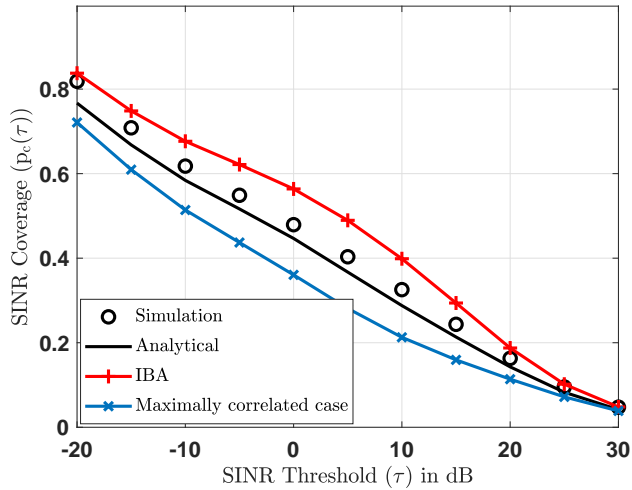} 
			\caption[Coverage distribution  for $\beta=3\beta_0,\mu = \mu_0$]{Coverage probability  of a 2D mmWave cellular network in presence of random blockage. Here, $\beta=.014/m$ and $\mu=220/km^2$. Values for IBA and maximally correlated case are generated substituting respective values of $\probserv{\L}(r)$ and 
				$\Lfunc_{\v|\s{}}(r,\tau)$ in \eqref{eq:PcFinalEq}.
			}
			\vspace{-.1in}
			\label{fig:SINR2}
	\end{center}
\end{figure}

\textbf{SINR distribution:}  
Fig. \ref{fig:SINR2} shows that the variation of the SINR coverage with respect to threshold along with various bounds. Here also, we can observe that considering blocking correlation gives better accuracy. At higher threshold, the different between the two expressions (the independent blocking and the maximally correlated case) may diminish as noise dominates. However, the difference is still significant when the SIR coverage of two cases is compared.

\section{Conclusions}
In this work, we evaluated the performance of a mmWave cellular network in the presence of blockages while considering the inherent blocking correlation present in the system. Owing to differences in the analysis of a 1D and 2D mmWave cellular network, we analyze them separately. We saw that the independent blocking assumption used in most prior work is generally inaccurate, especially when blockages tend to be large. We also investigated the impact of the blocking correlation on the LoS association probability and the coverage probability. The work has numerous possible extensions. First, the proposed analysis can also be extended to study other association rules and channel models. Second, the proposed framework can be utilized to model the correlation in the object detection probability of sensing and radar applications and investigate their performance.   
In dense urban environments with tall buildings, cellular systems have recently seen increasing BS deployments in the vertical direction as well
resulting in a nearly 3D deployment of BSs. Our analysis can be extended to analyze such 3D deployments.

%%%%%%%%%%%%%%%%%%%%%%%%%%%%%%%%%%%%%%%%%%%%%%%%%%%%%%%%%%%%%%%%%%%%%%%%%%%%%%%%

\appendices
\section{Proof of Theorem \ref{thm:GLr1D}}\label{app:thm:GLr1D}
The joint probability $G_{\L}(r)$ can be written as
\begin{align}
        G_{\L}(r)&=\expects{\Phi}{\sum\nolimits_{\X \in \Phi}
        \indside{E_{\L,\X}}
        \indside{\X\notin\Ball(\origin,r)}}\label{eq:11}\\      
        &
        \stackrel{(a)}{=} \lambda \int_{r}^\infty \palmprobx{
        E_{\L,\x}
         }{\x}\  \dd \x
        +\lambda \int_{-\infty}^{{-r}} \palmprobx{
        E_{\L,\x}
        }{\x}\  \dd \x \label{eq:GLR1D}\\
        &
        = \lambda \int_{r}^\infty g_\L(\x)  \dd \x
        +\lambda \int_{-\infty}^{{-r}} g_\L(\x)\dd\x
        \stackrel{(b)}{=}
        2\lambda \int_{r}^\infty g_\L(x)  \dd x
         \label{eq:GLR1D_1}
        \end{align}
where $(a)$ is due to the Campbell-Mecke theorem \cite{AndGupDhi16}, 
and 
$(c)$   is due to symmetry. 
Now, \eqref{eq:gLx1D_middlestep1} is due to the law of total probability. \eqref{eq:gLx} is obtained by
using \eqref{eq:cond_gL} in \eqref{eq:gLx1D_middlestep1}.

\section{Proof of \eqref{eq:cond_gL} and \eqref{eq:cond_gN}}
\label{appendix:proofth1}
To show \eqref{eq:cond_gL}, we note that 
\begin{align*}
g_\L(x \mid q, q')
 &= \prob {\text{The BS at distance } x \text{ is serving and LoS}}
 \\
 &=\indside{x<\qr}
 	\prob{|\y|>x\ \forall \y <\qr, \y>0\ } \ \prob{|\y|>x\ \forall |\y| <\ql,\y<0\ } \\
	&
	\qquad\qquad
 	\ \prob{|\y|> \excNL{}(x)\  \forall \y >\qr,\y>0\ }\ \prob{|\y|> \excNL{}(x) \ \forall |\y| >\ql,\y<0\ }.
\end{align*}
Now, using the PGFL of a PPP and $\excNL{}(x) < x$, we can write the above expression as,
\begin{align*}
g_\L(\x \mid q, q') = \mathbbm{1}(\x < q) \exp \left(-\lambda \int_{0}^{\min (q, x)} \dd y\right) \exp \left(-\lambda \int_{0}^{\min \left(q^{\prime}, x\right)} \dd y \right)  \exp \left(-\lambda \int_{q^{\prime}}^{\max(\excNL{}(x), q^{\prime})} \dd y\right).
\end{align*}
Solving it further, we get the desired result. Similarly to prove \eqref{eq:cond_gN}, we have 
\begin{align*}
g_\L(x \mid q, q')
 &=\indside{x>\qr}
 	\prob{|\y|>\excLN{}(|\x|)\ \forall |\y| <\qr, \y>0\ } \ \prob{|\y|>\excLN{}(|\x|)\ \forall |\y| <\ql,\y<0\ } \\
	&
	\qquad\qquad
 	\ \prob{|\y|> x\  \forall \y >\qr,\y>0\ }\ \prob{|\y|> x \ \forall |\y| >\ql,\y<0\ }.
\end{align*}
Now, using the PGFL of a PPP, we get
\begin{align*}
g_\N(x \mid q, q') = \mathbbm{1}(x > q) & \exp \left(-\lambda \int_{0}^{q} d y\right) \exp \left(-\lambda \int_{q}^{x} \dd y\right) \times \\
& \exp \left(-\lambda \int_{0}^{\min \left(q^{\prime}, \left(\excLN{}(x)\right) \right)} \dd y \right) 
 \exp \left(-\lambda \int_{q^{\prime}}^{\max(q^{\prime}, x)} \dd y\right).
\end{align*}
Solving further, we get the desired result.

\section{Derivation of Laplace Transform of Interference}\label{appendix:proof-interference}

Let us first compute   $\mathcal{L}_{I| E_{\L,\x},\ql,\qr}(s)$ which denotes the LT of interference conditioned on the serving BS at $\x$ is LoS and on the locations of nearest blockage. 
From Appendix \ref{appendix:proofth1}, we know that every other LOS BS $\y$ should satisfy $|\y|>x$ and every NLOS BS should satisfy
$
|\y| > \excNL{(x)}
$.

Now,  the sum interference $I(\Phi)$ from all BSs in $\Phi$ can be written as
\begin{align*}
I(\Phi)=I(\Phi_\L^+)+I(\Phi_\L^-)+I(\Phi_\N^+)+I(\Phi_\L^-)
\end{align*}
where individual terms denote the interference from the particular subset of the point process. We will compute the LT of individual terms. Due to their independence, the LT of $I$ will be the product of these LTs. Now, 
\begin{align*}
   \mathcal{L}_{I(\Phi_\L^+)| E_{\L,\x},\ql,\qr}(s)&=
   \expect{
   	\expS{ -s \sum\nolimits_{\Y_i \in \Phi_\L^{+}} H_i \ell_{\L}(|\Y_i|) 
	\indside{\Y_i > x}
	}
	}\\
	&= \expect{
   	\expS{ -s \sum\nolimits_{\Y_i \in \Phi} H_i \ell_{\L}(|\Y_i|) 
	\indside{\Y_i > x }\indside{\Y_i>0}\indside{\Y_i<\qr}
	}
	}.
  \end{align*}
   Now, from the Laplace functional of the marked PPP, 
   \begin{align*}
   \mathcal{L}_{I(\Phi_\L^+)| E_{\L,\x},\ql,\qr}(s)&
	= \expS{-\lambda \int_0^\infty 
	\indside{\qr>y>x}
	\left(1-\expects{H}{
   	\expS{ -s  H \ell_{\L}(y) 
	}}\right)\dd y
	}\\
	&=   {\expS{-\lambda \int_x^\qr 
	\left(1-\frac{1}{
   	1+{ s  \ell_{\L}( y) 
	}}\right)\dd y
	}}
  \end{align*}
  where the last step uses  the moment generating function (MGF) of $H\sim\mathsf{Exp}(1)$ and the fact that $x<q$. Similarly,
  \begin{align*}
   \mathcal{L}_{I(\Phi_\L^-)| E_{\L,\x},\ql,\qr}(s)&=
   \expect{
   	\expS{ -s \sum\nolimits_{\Y_i \in \Phi_\L^{-}} H_i \ell_{\L}(|\Y_i|) 
	\indside{|\Y_i| > {x} }
	}
	}\\
	&= \expect{
   	\expS{ -s \sum\nolimits_{\Y_i \in \Phi} H_i \ell_{\L}(|\Y_i|) 
	\indside{|\Y_i| > {x} }\indside{\Y_i<0}\indside{|\Y_i|<\ql}
	}
	}\\
	&={\expS{-\lambda \int_x^{\infty} 	
	\left(1-\frac{1}{
   	1+{ s \ell_{\L}(y)  
	}}\right)\indside{y<\ql}
	\dd y}}
  \end{align*}
Similarly,
\begin{align*}
   \mathcal{L}_{I(\Phi_\N^+)| E_{\L,\x},\ql,\qr}(s)&=
   \expect{
   	\expS{ -s \sum\nolimits_{\Y_i \in \Phi_\N^{+}} H_i \ell_{\N}(|\Y_i|) 
	\indside{\Y_i > \excNL{}(x)}
	}
	}\\
	&= \expect{
   	\expS{ -s \sum\nolimits_{\Y_i \in \Phi} H_i \ell_{\N}(|\Y_i|) 
	\indside{\Y_i > \excNL{}(x) }\indside{\Y_i>0}\indside{\Y_i>\qr}
	}
	}\\
	&={\expS{-\lambda \int_\qr^\infty 
	\left(1-\frac{1}{
   	1+{ s\ell_{\N}( y)  
	}}\right)\dd y
	}}.
  \end{align*}
  Here, we have used the fact that $\qr >\excNL{}(x)$. Also, 
  \begin{align*}
   \mathcal{L}_{I(\Phi_\N^-)| E_{\L,\x},\ql,\qr}(s)&=
   \expect{
   	\expS{ -s \sum\nolimits_{\Y_i \in \Phi_\N^{-}} g_i \ell_{\N}(|\Y_i|)  
	\indside{\Y_i > \excNL{}(x)}
	}
	}\\
	&= \expect{
   	\expS{ -s \sum\nolimits_{\Y_i \in \Phi} g_i \ell_{\N}(|\Y_i|)  
	\indside{\Y_i > \excNL{}(x) }\indside{\Y_i<0}\indside{|\Y_i|>\ql}
	}
	}\\
	&={\expS{-\lambda \int_{\excNL{}(x)}^\infty 
	\left(1-\frac{1}{
   	1+{ s \ell_{\N}(y) 
	}}\right)\indside{y>\ql}\dd y
	}}.
  \end{align*}
 Multiplying the four terms, we get the value of $ \mathcal{L}_{I| E_{\L,\x},\ql,\qr}(s)$.

The derivation for the term $ \mathcal{L}_{I(\Phi_\N^-)| E_{\L,\x},\ql,\qr}(s)$ is similar. One difference is due to the fact that 
if the user is receiving from an NLoS BS located at $\x>0$, it would mean that there are no LoS BS on the right side of the user. Hence, we will not include the LoS interference from $\Phi_\L^+$. 
%%%%%%%%%%%%%%%

\section{Proof of Theorem \ref{thm:Pc1D}}
\label{app:Pc1DProof}

The SINR coverage   as defined in \eqref{eq:pcdef} can be written as
\begin{align}
\Pc(\SThres) = &\expect{\sum\nolimits_{\X_i\in\Phi}\indside{\SINR>\SThres}\indside{\X_i \text{ is serving}}}\\
=&\expect{\sum\nolimits_{\X_i\in\Phi}\indside{\SINR>\SThres}\indside{E_{\L,\X_i}}}+\expect{\sum\nolimits_{\X_i\in\Phi}\indside{\SINR>\SThres}\indside{E_{\N,\X_i}}}.\label{eq:Pc1D_1}
\end{align}
Now, let us consider the first term in \eqref{eq:Pc1D_1}. Using the conditional probability law, we get
\begin{align}
&\expect{\sum\nolimits_{\X_i\in\Phi}\indside{\SINR>\SThres}\indside{E_{\L,\X_i}}}
=\expects{\ql,\qr}{\sum\nolimits_{\X_i\in\Phi}\prob{\SINR>\SThres|\ql,\qr,E_{\L,\X_i}}\prob{E_{\L,\X_i}|\ql,\qr}}\nonumber\\
&\stackrel{(a)}{=}2\expects{\ql,\qr}{\sum\nolimits_{\X_i\in\Phi\cap (0,\infty)}\prob{\SINR>\SThres|\ql,\qr,E_{\L,\X_i}}\prob{E_{\L,\X_i}|\ql,\qr}}\nonumber\\
&\stackrel{(b)}=2\lambda \int_0^\infty \prob{\SINR>\SThres|\ql,\qr,E_{\L,\x}}g_\L(\x|\ql,\qr) \dd\x,
\end{align}
where $(a)$ is due to the symmetry around the origin and $(b)$ is from the Campbell-Mecke theorem. Now, the inner term denoting the conditional distribution of SINR is 
\begin{align}
&\prob{\SINR>\SThres|\ql,\qr,E_{\L,\x}}=  \prob{\frac{H \ell_{\L}(x)}{\noise + I} > {\SThres}}
= \mathbb{P}\left( H > \frac{{\SThres} (\noise+ I)}{\ell_{\L}(x)} \right) \nonumber\\
&= \exp\left( - \frac{{\SThres} \noise }{\ell_{\L}(x)} \right)  
\expect{\exp\left( - \frac{{\SThres I}}{\ell_{\L}(x)} \right) |\ql,\qr,E_{\L,\x}}
= \exp\left( - \frac{\SThres \noise}{\ell_{\L}(x)} \right)  
\mathcal{L}_{I| E_{\L,\x},\ql,\qr}\left( \frac{\SThres}{ \ell_{\L}(x) } \right).\label{eq:86a}
\end{align}

Similarly for the second term in \eqref{eq:Pc1D_1}, we get
\begin{align}
&\expect{\sum_{\X_i\in\Phi}\indside{\SINR>\SThres}\indside{E_{\N,\X_i}}}
=2\lambda \int_0^\infty \prob{\SINR>\SThres|\ql,\qr,E_{\N,\x}}g_\N(\x|\ql,\qr) \dd\x,\\
&\text{with }\quad\ \ \,\prob{\SINR>\SThres|\ql,\qr,E_{\N,\x}} 
= \exp\left( - \frac{\SThres \noise}{\ell_{\N}(x)} \right) \cdot 
\mathcal{L}_{I| E_{\N,\x},\ql,\qr}\left( \frac{\SThres}{ \ell_{\N}(x)} \right)\label{eq:88a}.
\end{align}

Substituting \eqref{eq:1DLTIL} and \eqref{eq:1DLTIN} in \eqref{eq:86a} and \eqref{eq:88a} and then putting the resultant values in \eqref{eq:Pc1D_1}, we get desired result.

%%%%%%%%%%%%%%%%%
\section{}
\label{appendix:lsns1}
Let $B_{\x}$ denote the BS at  $\x{}$ from the user. Let $\|\x{}\|=x$.
\begin{align*}
    \probserv{\L}(\x) &= \palmprobx{B_{\x} \text{ is LoS serving BS}}{\x}
    =\palmexpectx{\indside{\text{$B_{\x}$ is LoS}}\indside{B_{\x}\text{ is serving}}}{\x}\\
     &=\palmprobx{{\text{$B_{\x}$ is LoS}}}{\x}\ \times \ \palmprobx{{B_{\x}\text{ is serving}}|{\text{$B_\x$ is LoS}}}{\x}
     =\probln{\L}({x})\ \palmprobx{{B_{\x}\text{ is serving}}|{\text{$B_{\x}$ is LoS}}}{\x}\\
& =\palmprobx{\prod_{\mathbf{t} \in \Phi} 
 \Big(
 \indside{
\dist{\t}\ge x
 } 
 \indside{B_\t \text{ is LoS}}
 +
 \indside{
\dist{\t}\ge \excNL{}(\dist{\x})
 } 
 \indside{B_\t \text{ is NLoS}}\Big)|B_\x\text{ is LoS }}{\x}
 \\
    & = \exp
    \left(-\lambda \int_0^{x} \int_0^{2\pi}\prob{\text{BS at $(t,\theta)$ is LoS}|\text{LoS }B_\x}\dd \theta t \dd t
     \right. \nonumber \\
    & \left.\qquad \qquad
      -\lambda \int_0^{\excNL{}(\dist{\x})} \int_0^{2\pi} 
    \prob{\text{BS at $(t,\theta)$ is NLoS}|\text{LoS } B_{\x}}
    \dd\theta t \dd t\right)
    \end{align*}
where the last step is due to use of polar coordinates $\t=(t,\theta)$. Using the conditional probability of LoS and NLoS, we get the final result.

\section{Proof of Lemma \ref{lemma:glxgeneralbounds} }
\label{app:glxgeneralbounds}

From \eqref{eq:glx} , we get
\begin{align}
&\probserv{\L}(x) = \probln{\L}({x}) 
\expS{-\lambda\int_0^{2\pi} \left[ \int_0^{{x}}\frac{\probln{\L,\L}({x},t,\theta)}{\probln{\L}({x})}t\dd t \dd \theta-
 \int_0^{\excNL{}({x})}
\frac{\probln{\L,\N}({x},t,\theta)}
{\probln{\L}({x})}
t\dd t \right] \dd\theta}\nonumber\\
&= \probln{\L}({x}) 
\expS{-\lambda\int_0^{2\pi} \left[
\int_{\excNL{}({x})}^{{x}}\frac{\probln{\L,\L}({x},t,\theta)}{\probln{\L}({x})}t\dd t \dd \theta-
 \int_0^{\excNL{}({x})}
\frac{\probln{\L,\L}({x},t,\theta)+\probln{\L,\N}({x},t,\theta)}
{\probln{\L}({x})}
t\dd t \right]
\dd\theta}\nonumber\\
&= \probln{\L}({x}) 
\expS{-\lambda\int_0^{2\pi} 
\int_{\excNL{}({x})}^{{x}}\frac{\probln{\L,\L}({x},t,\theta)}{\probln{\L}({x})}t\dd t \dd \theta}
\expS{-
\lambda \pi {\excNL{}^2({x})}
}.\label{app:eq:gLSimpleExp}
\end{align}
It can be seen that $g_\L$ decreases with increase in $\probln{\L,\L}({x})$, which means that it decreases with blockage correlation. 
Now, 
 \eqref{eq:G_Ind_Gen} can be obtained by substituting the bound on $\probln{\L,\L}(r,t,\theta)$ as given in \eqref{eqn:joint_prob_lb} in  \eqref{app:eq:gLSimpleExp}. We substitute \eqref{eqn:joint_prob_ub} in \eqref{eq:glx} to get \eqref{eq:G_Inf_p1_Gen}. Now using the fact that 
$$\int_0^{\pi/2} e^{-a \sin^2(\theta)}\dd\theta=\pi/2 \expU{-a/2}I_0(a/2),$$
and further integrating it, we get \eqref{eq:G_Inf_Gen}.

\section{Proof of Theorem \ref{thm:laplaceI}}\label{app:thm:laplaceI}
\vspace{-.3in} 
\begin{align}
 &\laplace{I_\v|X_0=r,E_\u}(\nu)=\expect{\expS{-\nu \sum\nolimits_{\X_i\in\Phi_\v, \X_i\ne \X_0}H_i 
 \ell_\v{(X_i)}
 }\mid X_0=r,E_\u }\nonumber\\
  &=\expect{\expS{-\nu\!\! \sum\nolimits_{\X_i\in\Phi, \X_i\ne \X_0}H_i 
 \ell_\v{(X_i)}   
  \indside{\X_i \text{'s type is } \v}
    }\right. \left. \mid X_0=r,E_\u \vphantom{\frac12}}\nonumber\\
     &=\expect{
     \exp\left(
     -\nu \sum\nolimits_{\X_i\in\Phi\setminus\Ball(\origin,\exc{\v|\u}(r)) }
     H_i
      \ell_\v{(X_i)}   \ 
    \prob{\X_i \text{'s type is } \v \mid \text{BS at $r$ is of type } \u}     \right)}\nonumber.
    \end{align}
Now, using the PGFL of PPP \cite{HaenggiBook} and MGF of $H_i$, we get
\begin{align}
 \laplace{I_\v|R_\u=r,E_\u}(\nu)=  &\expS{-\frac{%2 \pi
   \lambda}{\probln{\u}(r)} 
  \int_{\exc{\v|\u}(r)}^{\infty}
  \frac{\nu 
  \ell_\v(t)
  			}
  {1+\nu
   \ell_\v(t) 
  } 
\int_0^{2\pi}
{\probln{\u,\v}(r,t,\theta)} d\theta t \dd t }\nonumber.
\end{align}

\section{Proof of Lemma \ref{lemma:simpleboundgL2D}}
\label{appendix:sin1}

For \eqref{eq:G_lb} and \eqref{eq:G_ub}, we first derive the following novel linear bound on $\sin^2 \theta $,  given as
\begin{align}
\label{eqn:sin_sqrd_bound}
    1+m(x-\frac{\pi}{2})\leq \sin^2 x\leq mx
\end{align}
for   $\theta \in [ 0,\pi/2 ] $ with $m=1.38$. The proof is as follows.

Let us take $f(x) = \sin^2x-mx$ where $m=1/1.38$.
On maximizing $f(x)$ for $x \in [0,\pi/2]$ using first derivative condition, we get
 $$\max_{x \in [0,\pi/2]} f(x) =  \frac{1}{2} (1+\sqrt{1-m^2}-m(\pi- \sin^{-1}(m))).$$
We can observe that for $m=1/1.38$,
 $$\max_{x \in [0,\pi/2]} f(x) = -0.00006.$$
 This means that 
$ f(x) < 0 \ \forall \ x \in [0,\pi/2]$, which proves the upper bound. Similarly, we can prove the lower bound. 
We can also verify from Fig. \ref{fig:app:sinsqbound} that $\sin^2x$ is closely  bounded by the lines $y=x/1.38$ and $y=x/1.38+1-\pi/2.76$.
\begin{figure}[ht!]
\begin{center}
\includegraphics[width=.4\textwidth]{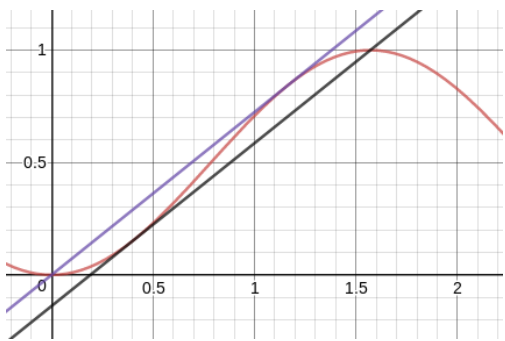}
\caption{Illustration showing linear bounds on $\sin^2x$ for $x \in [0,\frac{\pi}{2}]$}
\label{fig:app:sinsqbound}
\end{center}

\end{figure}

\vspace{.2in}

\bibliographystyle{IEEEtran}
\bibliography{thesis} 

\end{document}